\documentclass[12pt]{article}

\usepackage{amsmath,amssymb}

\usepackage{times}
\usepackage{bm}
\usepackage{natbib}
\usepackage{algorithm}
\usepackage{xargs}[2008/03/08]
\usepackage{amssymb}
\usepackage{mathrsfs}
\usepackage{graphicx}
\usepackage{rotating,subfigure}
\usepackage[flushleft]{threeparttable}
\usepackage{multirow}
\usepackage{enumitem} 
\usepackage{subfigure}

\usepackage{star}

\usepackage[colorlinks,
linkcolor=blue,
anchorcolor=blue,
citecolor=blue
]{hyperref}

\def\T{{ \mathrm{\scriptscriptstyle T} }}

\def\##1\#{\begin{align}#1\end{align}}
\def\$#1\${\begin{align*}#1\end{align*}}

\newcommand{\rF}{\textnormal{F}}

\def\T{{ \mathrm{\scriptscriptstyle T} }} 

\newcommand{\Rom}[1]{\text{\uppercase\expandafter{\romannumeral #1\relax}}}

\newcommand{\scolor}[1]{{\color{black}#1}}

\newcommand{\minimize}{\mathop{\mathrm{minimize}}}

\usepackage{txfonts}

\usepackage{geometry}
\begin{document}

\title{ \LARGE  Resistant convex clustering: How does the fusion penalty enhance resistance?}

\author{
 Qiang Sun\thanks{University of Toronto and MBZUAI; E-mail: \texttt{qsunstats@gmail.com}.}
\and Archer Gong Zhang\thanks{University of Toronto; E-mail: \texttt{archer.zhang@utoronto.ca}.}
\and Chenyu Liu\thanks{University of California, San Diego; E-mail: \texttt{chl056@ucsd.edu}.}
\and Kean Ming Tan\thanks{University of Michigan, Ann Arbor; E-mail: \texttt{keanming@umich.edu}.}
 }

\date{}

\maketitle


\begin{abstract}  
Convex clustering is a convex relaxation of the $k$-means and hierarchical clustering. It involves solving a convex optimization problem with the objective function being a squared error loss plus a fusion penalty that encourages the estimated centroids for observations in the same cluster to be identical. 
However, when data are contaminated, convex clustering with a squared error loss fails even when there is only one arbitrary outlier.
To address this challenge, we propose a resistant convex clustering method.
Theoretically, we show that the new estimator is resistant to arbitrary outliers: it does not break down until more than half of the observations are arbitrary outliers.
Perhaps surprisingly, the fusion penalty can help enhance resistance by fusing the estimators to the cluster centers of uncontaminated samples, but not the other way around. 
Numerical studies demonstrate the competitive performance of the proposed method.  
\end{abstract}
\textbf{Keywords:} Breakdown point,  fusion penalty, outliers, resistance, robustness.

\section{Introduction}


Clustering is ubiquitous in many scientific disciplines such as pattern recognition, machine learning, and bioinformatics.  
Given $n$ observations, the goal of clustering is to group the $n$ observations into $k$ clusters.
Traditional clustering algorithms such as $k$-means and hierarchical clustering take a greedy approach and are sensitive to initializations of the clusters, and the choice of distance metric and linkage, respectively \citep{ElemStatLearn,johnson2002applied}, due to their non-convex nature.

To avoid the non-convexity issue, several authors have proposed a convex formulation of the clustering problem, referred to as convex clustering \citep{pelckmans2005convex,hocking2011clusterpath,lindsten2011clustering}. 
Specifically, convex clustering solves a convex optimization problem with the cost function being  a squared error loss plus a fusion penalty that encourages the centroids  of observations in the
same cluster to be identical.
Efficient algorithms for convex clustering have been developed \citep{chi2013splitting,chen2014convex,sun2018convex,Weylandt2020dynamic}.   
Theoretical properties of convex clustering were studied \citep{NIPS2014_5307,tan2015,wang2016sparse,radchenko2014consistent,chi2018recovering}.
\citet{chi2014convex} and \cite{chi2018provable} considered extensions of convex clustering to perform co-clustering on matrices and tensors. 

Convex clustering is developed based on an inherent assumption that there are no outliers in the data.
However, in practice, large-scale data sets are often corrupted.  
Due to the use of squared error loss, a naive application of convex clustering will cluster each outlier into a singleton cluster. 
To address this issue, we propose a resistant convex clustering method by substituting the squared error loss in the convex clustering formulation with a Huber loss \citep{huber1964, huber1973}. The resulting optimization problem is convex, which we solve using an alternative direction method of multipliers algorithm.  We refer readers to  \cite{rousseeuw1984least, rousseeuw1984robust, yohai1987high, mizera1999breakdown} and \cite{salibian2002bootrapping} for classical analysis of resistant $M$-estimators in the presence of arbitrary outliers, and to   \cite{catoni2012challenging, sun2018adaptive, avella2018robust, ke2018user,tan2018robust} for nonasymptotic analysis of Huber regression with a diverging robustness parameter under heavy-tailed distributions. 

We analyze the breakdown point of the proposed resistant convex clustering method.
Informally, the breakdown point of an estimator is defined as the proportion of arbitrary outliers an estimator can tolerate before the estimator produces arbitrarily large estimates or breaks down \citep{hampel1971general}. 
We show that the proposed estimator does not break down until more than half of the observations are arbitrary outliers. 
This is perhaps rather surprising, at least to us,  as we expected one arbitrary large outlier will destroy the clustering procedure because there are as many parameters as the samples. Comparing with the estimator without the fusion penalty, we find that the fusion penalty helps enhance the resistance of the clustering procedure by fusing the estimators of the centroids to the cluster centroids of uncontaminated observations, but not the other way around. To the best of our knowledge, such phenomenon has not yet been observed in the literature. 
The R package that implements our method can be found at \url{https://github.com/statsle/Rcvxclustr}. 


\paragraph{Related work} We review related work on robust clustering methods. 
Existing robust clustering methods include 
the trimmed $k$-means \citep{garcia2010review,whang2015non} and robust mixture models \citep{peel2000robust,lin2007robust}. The trimmed $k$-means algorithm first picks an outlying proportion and then optimizes the trimmed $k$-means objective \citep{cuesta1997trimmed}, and is well developed \citep{garcia1999robustness, georgogiannis2016robust,dorabiala2022robust}.
However, the trimmed $k$-means tends to produce clusters with the same size \citep{garcia2010review}, and may fail dramatically when the clusters are unbalanced. 
The robust mixture models further mitigate the cluster unbalanced issue by explicitly modeling the marginal clustering probabilities \citep{gallegos2005robust, cuesta2008robust, yang2012robust,mclachlan2019finite}. 
However, both the trimmed $k$-means and robust mixture models are non-convex optimization problems, and thus finding the global optima is challenging.




\section{Resistant Convex Clustering}
\label{rcc}
Let $\Xb \in \RR^{n \times p}$ be a data matrix with $n$ observations and $p$ features. 
A popular variant of convex clustering estimates a centroid matrix $\Ub\in\RR^{n\times p}$ by solving the following convex optimization problem
\begin{equation}
\label{eq:ccl}
 \widehat\Ub^{\rm ls}(\lambda) = \underset{\Ub \in \RR^{n\times p}}{\mathrm{argmin}}~ \frac{1}{2} \sum_{i=1}^n \|\Xb_{i}- \Ub_{i}\|_2^2  + \lambda \sum_{i<i'}\|\Ub_{i}-\Ub_{i'} \|_2,
\end{equation}
where $\Xb_{i}$ and $\Ub_i$ are the $i$th row of $\Xb$ and $\Ub$ respectively, \scolor{ and $\lambda \geq 0$ is a tuning parameter} \citep{pelckmans2005convex,hocking2011clusterpath,lindsten2011clustering}.  When it is clear from the context, we omit $\lambda$ and write $\widehat\Ub^{\rm ls}(\lambda)$ as $\widehat\Ub^{\rm ls}$. 
In order to distinguish \eqref{eq:ccl} with its resistant version to be developed later, we refer to \eqref{eq:ccl} as the least-squares convex clustering problem. 
The fused group lasso penalty, $\|\Ub_{i}-\Ub_{i'} \|_2$, encourages the rows of $\hat{\Ub}^{\rm ls}$ to be similar to each other.
The number of unique rows in $\hat{\Ub}^{\rm ls}$ is controlled by 
the nonnegative tuning parameter $\lambda$.
The cluster assignments can be inferred based on $\hat{\Ub}^{\rm ls}$: the $i$th and $i'$th observations are estimated to belong to the same cluster if and only if $\hat{\Ub}^{\rm ls}_{i}=\hat{\Ub}^{\rm ls}_{i'}$.

Because the squared error loss is sensitive to outliers, the least-squares convex clustering often fails to identify the correct cluster memberships when data are contaminated. Indeed, we have the following informal result indicating the least-squares convex clustering is not resistant to arbitrary data contamination, with its formal version presented in Section \ref{sec:analysis}. 
\begin{theorem}[Informal Statement]
The least-squares convex clustering breaks down when there is only $1$ arbitrary bad data point. 
\end{theorem}

The above theorem states that if one single observation is adversarially contaminated to take an arbitrary value, then the least-squares convex clustering breaks down. To address this issue, we propose to substitute the squared error loss in \eqref{eq:ccl} by a loss function such that the resulting procedure is resistant to outliers. In Section \ref{sec:analysis}, we will show that the Huber loss combined with the fused group lasso penalty is resistant to outliers in terms of breakdown point analysis, where 
the Huber loss  is formally defined as \citep{huber1964}:  
\begin{equation}
\label{eq:ccl2}
	\ell_\tau(a) =
	\left\{\begin{array}{ll}
	\frac{1}{2}x^2 ,    & \mbox{if } |x | \leq \tau ,  \\
	\tau |x | - \frac{1}{2} \tau^2,   &  \mbox{if }  |x | > \tau
	\end{array}  \right.
\end{equation}
where \scolor{ $\tau>0$} is a cutoff parameter that blends the quadratic region and the linear region of the loss function. 
The Huber loss induces resistance since it grows slower,  linearly instead of quadratically,  at tails where $|x|>\tau$. 
We then propose to estimate the centroid matrix $\Ub$ by
\begin{equation}
\label{eq:ccl3.1}
\widehat \Ub = \underset{\Ub \in \RR^{n\times p}}{\mathrm{argmin}}~  \sum_{i=1}^n \ell_\tau (\Xb_{i}- \Ub_{i})  + \lambda \sum_{i<i'}\|\Ub_{i}-\Ub_{i'} \|_2,
\end{equation}
where we use the notation  $\ell_\tau (\Xb_{i}- \Ub_{i})$ to indicate $\sum_{j=1}^p \ell_\tau(X_{ij} - U_{ij})$.
Note that \eqref{eq:ccl3.1} reduces to the least-squares convex clustering  \eqref{eq:ccl} when $\tau \rightarrow \infty$.
The optimization problem \eqref{eq:ccl3.1} is convex, and thus an efficient algorithm can be developed to achieve the global optimum.

\section{Breakdown Point Analysis}
\label{sec:analysis}
In this section, we examine the breakdown point property of the least-squares convex clustering method and our proposed estimator. 
Recall that $\Xb \in \RR^{n \times p}$ is the original data that are uncontaminated. 
We define the set 
\$
\cP_m(\Xb)=\big\{\widetilde\Xb:\,  \tilde \Xb_{i}\ne  \Xb_{i},~ i\in \cI~\mathrm{such~that~} |\cI|\leq m\big\}.
\$
In other words, $\cP_m(\Xb)$ is the set of all possible contaminated data matrices that are obtained by replacing at most $m$ rows of the original data $\Xb$, which we refer to as the $m$-row-wise contamination model. 
Throughout this section, let $\tilde{\Xb}\in\cP_m(\Xb)$ be the contaminated data.
Let $\hat{\Ub}(\Xb)$ and $\hat{\Ub}(\tilde{\Xb})$ be the solutions to \eqref{eq:ccl3.1} with the original data $\Xb$ and the contaminated data $\tilde{\Xb}$, respectively.
We now provide a formal definition of the breakdown point of any estimator \citep{donoho1983notion}.

\begin{definition}
\label{def:breakdown}
 The breakdown point of an estimator  $\hat{\Wb}$ is defined as 
\[
\varepsilon^*(\widehat\Wb,  \Xb)=\min\left\{\frac{m}{n}: \sup_{\tilde\Xb\in \cP_m(\Xb)}\big\|\widehat\Wb (\tilde\Xb)-\hat{\Wb} (\Xb)\big\|_\rF=\infty\right\}.
\]
\end{definition}
The supremum is taken over all possible contaminated datasets in $\cP_m (\Xb)$.  
Thus, the quantity $\varepsilon^*(\hat{\Wb},\Xb)$ can be interpreted as the smallest proportion of contaminated samples for which the estimator $\hat{\Wb}$ produces an arbitrarily large estimate relative to $\widehat\Wb(\Xb)$.  
Our first result gives the breakdown point of the least-squares convex clustering estimator, \scolor{with its proof collected in Appendix~\ref{appendix:b0}.}


\begin{theorem}\label{thm:ls}
The breakdown point of the least-squares convex clustering estimator $\widehat \Ub^{\rm ls}(\lambda)$ is $1/n$ for any $\lambda \geq 0$. 
\end{theorem}

Theorem~\ref{thm:ls} indicates that the least-squares convex clustering is not resistant to adversarial contamination. 
In particular, if one single observation is adversarially contaminated to take an arbitrary value, then the least-squares convex clustering completely breaks down. On the contrary, our proposed method achieves a breakdown point of at least one half, which is formally presented below.   

\begin{theorem}
\label{thm:bp}
Take 
$
\tau < 
\lambda ({n-\lfloor(n+1)/2\rfloor})/{\sqrt{p}},
$ 
\scolor{ where $\lfloor \cdot \rfloor$ is the floor function.}
Then 
the resistant convex clustering estimator obtained from solving~\eqref{eq:ccl4} has a breakdown point of at least $1/2$ and at most $\lfloor(n+1)/2\rfloor/n$, 
that is
\$
\frac{1}{2}\leq \varepsilon^*(\widehat\Ub,  \Xb) \leq \frac{\lfloor (n+1)/2\rfloor}{n}. 
\$
\end{theorem}


\begin{proof}[Proof of Theorem \ref{thm:bp}]
The proof for the upper bound is standard and we collect it in the appendix for completeness. 
We only prove the lower bound here.  
Let
\begin{equation}
\label{eq:cclproof}
\cL_\tau(\Ub, \Xb)=   \sum_{i=1}^n \ell_\tau(\Xb_{i}-\Ub_{i})   + \lambda \sum_{i<i'}\|\Ub_{i}-\Ub_{i'} \|_2,
\end{equation}
where $\ell_\tau(\cdot)$ is the Huber loss defined in \eqref{eq:ccl2}.   
Recall from Definition~\ref{def:breakdown} the breakdown point of an estimator,  $\varepsilon^*(\widehat\Ub, \Xb)$.
Let $m=n\varepsilon^*(\widehat\Ub, \Xb)$. 
For every $k\in\NN$, there exists an $\tilde\Xb^k\in \cP_m(\Xb)$ such that $\|\widehat\Ub (\tilde\Xb^k)-\hat{\Ub}(\Xb)\|_{\rF}>k$, where $\hat\Ub (\tilde\Xb^k)$ and $\hat{\Ub}(\Xb)$ are estimators obtained from minimizing $\cL_\tau(\Ub, \tilde{\Xb}^k)$ and $\cL_\tau(\Ub, \Xb)$, respectively. 
Without loss of generality, we assume that the first $n-m$ samples in $\tilde\Xb^k$ are uncontaminated, i.e., $\tilde\Xb^k=(\Xb_{1}, \ldots, \Xb_{n-m},  \Yb^k_{n-m+1}, \dots, \Yb^k_{n})$, where $\Yb^k$ are the contaminated data.
For notational simplicity, we write $\Ub ^k=\widehat\Ub (\tilde\Xb^k)$. 
Moreover, we define two sets that contain indices for the good data points and contaminated data points, $\cG=\big\{1,\dots, n-m\big\}$ and $\cG^c=\big\{n-m+1,\dots, n\big\}$, respectively. 

Since $\Ub^k$ is the minimizer of $\cL_\tau(\Ub, \tilde\Xb^k)$, we have $ \cL_\tau(\Ub^k,\tilde\Xb^k)\leq \cL_\tau\big(\mathbf{0}, \tilde\Xb^k\big)$, implying 
\#\label{thm2:eq1}
\sum_{i\in \cG}\ell_\tau\big(\Xb_{i}- \Ub_i^k\big)+\sum_{i\in \cG^c}\ell_\tau\big(\Yb_{i}^k-\Ub_i^k\big)+\lambda \sum_{i<i'}\|\Ub_{i}^k-\Ub_{i'}^k\|_2 
&\leq \sum_{i\in \cG}\ell_\tau(\Xb_{i})+\sum_{i\in \cG^c}\ell_\tau(\Yb_{i}^k). 
\#
By Lemma~\ref{lemma:bound} and the symmetry of Huber loss, we obtain
\begin{equation}
\label{thm2:eq1-1}
\sum_{i\in \cG}\ell_\tau\big(\Ub_i^k\big) \le \sum_{i\in \cG} \ell_\tau\big(\Xb_i-\Ub_i^k\big)   + \sum_{i\in \cG}\ell_\tau\big(\Xb_i\big) +(n-m)p\tau^2
\end{equation}
and 
\begin{equation}
\label{thm2:eq1-2}
\sum_{i\in \cG^c}\ell_\tau\big(\Yb_i^k\big) \le \sum_{i\in \cG^c} \ell_\tau\big(\Yb_i^k-\Ub_i^k\big)   + \sum_{i\in \cG^c}\ell_\tau\big(\Ub_i^k\big) +mp\tau^2.
\end{equation}
Substituting \eqref{thm2:eq1-1} and \eqref{thm2:eq1-2} into \eqref{thm2:eq1} yields
\begin{equation}
\label{thm2:eq1-3}
\sum_{i\in \cG}\ell_\tau(\Ub_i^k)-2\sum_{i\in \cG}\ell_\tau(\Xb_{i})- \sum_{i\in \cG^c}\ell_\tau(\Ub_i^k)-np\tau^2+\lambda \sum_{i<i'}\|\Ub_{i}^k-\Ub_{i'}^k\|_2\leq 0.
\end{equation}

We now study the effect of the fused group lasso penalty on the contaminated data.
The penalty term can be rewritten as
\$
 \sum_{i<i'}\|\Ub_{i}^k-\Ub_{i'}^k\|_2
 =\sum_{i\in \cG,i
 '\in\cG^c}\big\|\Ub^k_i-\Ub^k_{i'}\big\|_2+\sum_{i<i':i,i' \in \cG} \big\|\Ub^k_i-\Ub^k_{i' }\big\|_2+\sum_{i<i': i,i' \in \cG^c} \big\|\Ub^k_i-\Ub^k_{i'}\big\|_2.
\$
By definition, as $k\rightarrow\infty$, $\sum_{i=1}^n\|\Ub_i^k\|_2\rightarrow \infty$. 
Per the compactness of the closed unit ball, we may assume that $\bm{\theta}^k_i=\Ub_i^k/\sum_{i=1}^n\|\Ub_i^k\|_2$ converges to some point $\bm{\theta}_i^0$, passing to a subsequence otherwise.

Dividing \eqref{thm2:eq1-3} by $\ell_\tau(\sum_{i=1}^n\|\Ub^k_i\|_2)$ and taking the limit when $k\rightarrow \infty$, we obtain
\#
&\liminf_{k\rightarrow \infty}\frac{\sum_{i\in\cG}\ell_\tau\left(\Ub_i^k\right)}{\ell_\tau(\sum_{i=1}^n\|\Ub^k_i\|_2)}
+ \lambda \liminf_{k\rightarrow \infty} \frac{\sum_{i\in\cG,i'\in\cG^c}\|\Ub_i^k-\Ub_{i'}^k\|_2}{\ell_\tau(\sum_{i=1}^n\|\Ub^k_i\|_2)}\notag\\
&\qquad +\lambda \liminf_{k\rightarrow \infty} \frac{\sum_{i<i': i, i' 
\in\cG\, {\rm or}\, i, i' 
\in\cG^c  }\|\Ub_i^k-\Ub_{i'}^k\|_2 }{\ell_\tau(\sum_{i=1}^n\|\Ub^k_i\|_2)} - \limsup_{k\rightarrow \infty}\frac{\sum_{i \in \cG^c}\ell_\tau\left(\Ub_i^k\right)}{\ell_\tau(\sum_{i=1}^n\|\Ub^k_i\|_2)} \leq 0.
\label{thm:normal}
\#
Dropping the third term on the left hand side and by Lemmas \ref{lemma:bound2}--\ref{lemma:bound3}, \eqref{thm:normal} reduces to
\$
\sum_{i\in \cG} \big\|\bm{\theta}_i^0\big\|_1+\frac{\lambda}{\tau}\sum_{i\in\cG, i'\in\cG^c}\big\|\bm{\theta}_i^0-\bm{\theta}_{i'}^0\big\|_2 -\sum_{i\in\cG^c}\big\|\bm{\theta}_{i}^0\big\|_1\leq 0.
\$
Using the fact that $\|\bm{z}\|_2\geq {p}^{-1/2}\|\bm{z}\|_1$ for any vector $\bm{z}\in \RR^p$, we obtain 
\begin{equation}
\label{eq:imp1}
\sum_{i\in \cG} \big\|\bm{\theta}_i^0\big\|_1+\frac{\lambda}{\tau\sqrt{p}}\sum_{i\in\cG, i'\in\cG^c}\big\|\bm{\theta}_i^0-\bm{\theta}_{i'}^0\big\|_1 -\sum_{i\in\cG^c}\big\|\bm{\theta}_{i}^0\big\|_1\leq 0.
\end{equation}

We now analyze \eqref{eq:imp1} by considering two cases. For the first case, 
by the triangle inequality $\| \btheta_i^0 - \btheta_{i'}^0\|_1 \ge \|\btheta_i^0\|_1 - \| \btheta_{i'}^0\|_1$, \eqref{eq:imp1} reduces to 
\$
\left(1+\frac{m\lambda}{\tau\sqrt{p}}\right)\sum_{i \in\cG}\big\|\bm{\theta}_i^0\big\|_1-\left(1+\frac{(n-m)\lambda}{\tau\sqrt{p}}\right)\sum_{i\in \cG^c}\big\|\bm{\theta}_i^0\big\|_1\leq 0.
\$
Simplifying the above expression yields
\begin{equation}
\label{eq:imp2}
A_{+}\sum_{i\in \cG}\big\|\bm{\theta}_i^0\big\|_1-\sum_{i\in \cG^c}\big\|\bm{\theta}_i^0\big\|_1\leq 0,\;\operatorname{where}\;A_{+}=\frac{m\lambda/(\tau\sqrt{p})+1}{(n-m)\lambda/(\tau\sqrt{p})+1}.
\end{equation}
For the second case, we use the triangle inequality 
$\| \btheta_i^0 - \btheta_{i'}^0\|_1 \ge \|\btheta_{i'}^0\|_1 - \| \btheta_{i}^0\|_1$. 
Following a similar calculation, we obtain
\$
\left(\frac{(n-m)\lambda}{\tau\sqrt{p}}-1\right)\sum_{i\in \cG^c}\big\|\bm{\theta}_i^0\big\|_1-\left(\frac{m\lambda}{\tau\sqrt{p}}-1\right)\sum_{i\in \cG}\big\|\bm{\theta}_i^0\big\|_1\leq 0.
\$
The above inequality can be simplified to
\begin{equation}
\label{eq:imp3}
\sum_{i\in \cG^c}\big\|\bm{\theta}_i^0\big\|_1-A_{-}\sum_{i\in \cG}\big\|\bm{\theta}_i^0\big\|_1\leq 0,\;\operatorname{where}\;A_{-}=\frac{m\lambda/(\tau\sqrt{p})-1}{(n-m)\lambda/(\tau\sqrt{p})-1},
\end{equation}
provided that the tuning parameters $\tau$ and $\lambda$ are chosen such that ${(n-m)\lambda}>{\tau\sqrt{p}}$.

Combining \eqref{eq:imp2} and \eqref{eq:imp3}, we obtain 
\$
A_{+}\sum_{i\in\cG}\big\|\bm{\theta}_i^0\big\|_1-A_{-}\sum_{i\in \cG}\big\|\bm{\theta}_i^0\big\|_1\leq 0
\$
provided ${(n-m)\lambda}>{\tau\sqrt{p}}$. 
Now if $\sum_{i\in \cG}\|\bm{\theta}_{i}^{0}\|_1\neq0$, we immediately have $A_{+}\leq A_{-}$, which further implies
\$
\frac{m}{n}\geq 1/2. 
\$
If $\sum_{i\in\cG}\|\bm{\theta}_{i}^{0}\|_1=0$, by \eqref{eq:imp3}, we must have $\sum_{i \in \cG^c}\|\bm{\theta}_{i}^{0}\|_1=0$. 
However, this contradicts the fact that $\sum_{i=1}^n\|\bm{\theta}^0_i\|_2=1$ by construction. 

Therefore, given $\tau/\lambda\leq  \frac{n-m}{\sqrt{p}}$,  we obtain 
\#\label{bdp:upper}
\frac{1}{2}\leq \frac{m}{n}\leq 1-\frac{\tau \sqrt{p}}{n\lambda}. 
\#
The  statement follows by observing $m=n\varepsilon^*(\widehat\Ub, \Xb)$. 
\end{proof}

The theorem above implies that as long as $\tau$ is not too large, our proposed resistant convex clustering has a breakdown point of at least $1/2$.    Let us define the Huber regression estimator without the fusion penalty as 
\$
\widehat\Ub^{\rm wo} = \underset{\Ub \in \RR^{n\times p}}{\mathrm{min}}~ \sum_{i=1}^n \ell_\tau(\Xb_i-\Ub_i).   
\$ 

Our last result states that $\widehat\Ub^{\rm wo}$ has a breakdown point of at most $1/n$,   \scolor{ with the  proof provided in Appendix~\ref{appendix:c}.}

\begin{proposition}\label{prop:wo_penalty}
The breakdown point of the Huber regression estimator without the fusion penalty $\widehat\Ub^{\rm wo}$ is  $1/n$.  
\end{proposition}

Perhaps surprisingly, at least to us, the two results above suggest that the fusion penalty helps to improve the resistance of the Huber estimator by fusing it to the empirical cluster centers of uncontaminated data points, but not the other way around.  Specifically, the fusion penalty helps enhance the resistantness property by improving the breakdown point from  $1/n$ for the Huber estimator to at least  $1/2$ for the resistant convex clustering estimator.   On the contrary, when the loss function is the least-squares loss, which is not as robust as the Huber loss, the fusion penalty does not help to improve resistance; see Theorem \ref{thm:ls}.   


\section{Implementation}
This section develops an algorithm for solving a general version of \eqref{eq:ccl3.1}: 
\#\label{eq:ccl3}
\underset{\Ub \in \RR^{n\times p}}{\mathrm{minimize}}~  \sum_{i=1}^n \ell_\tau (\Xb_{i}- \Ub_{i})  + \lambda \sum_{i<i'}w_{ii'}\|\Ub_{i}-\Ub_{i'} \|_2,
\#
where we allow general \scolor{ nonnegative weights $w_{ii'} \geq 0$}. Since \eqref{eq:ccl3} is a convex optimization problem, we solve \eqref{eq:ccl3} using an alternating direction method of multipliers algorithm \scolor{ (ADMM)} \citep{boyd2004convex}. 
Our algorithm is a modified version of that of \citet{chi2013splitting} to accommodate the Huber loss. 
The main idea is to decouple the terms in \eqref{eq:ccl3} that are difficult to optimize jointly.  
Let $\Vb$ be an ${n \choose 2} \times p$ matrix.  With some abuse of notation, let $\Vb_{ii'}$ be the row of $\Vb$ corresponding to the pair of indices $(i,i')$.
We recast \eqref{eq:ccl3} as the following equivalent constrained problem
\begin{equation}\label{eq:ccl4}
\begin{split}
&\minimize_{\Ub,\Wb\in \RR^{n\times p},\Vb\in \RR^{{n\choose 2}\times p}}~ 
~\sum_{i=1}^n \ell_\tau (\Xb_{i}- \Wb_{i})  + \lambda \sum_{i<i'}w_{ii'}\|\Vb_{ii'} \|_2\\
&\mathrm{subject~to}~~~ \Ub_{i} = \Wb_{i}, ~~~\Ub_{i} -\Ub_{i'} = \Vb_{ii'}, ~~ \forall~ i<i'.
\end{split}
\end{equation}
Construct an ${n\choose 2}\times n$ matrix $\Eb$ such that $(\Eb \Ub)_{ii'} = \Ub_i-\Ub_{i'}$.
Then, it can be shown that the scaled augmented Lagrangian function for \eqref{eq:ccl4} takes the form
\begin{align*}
L_{\tau}(\Wb,\Vb,\Ub,\Yb,\Zb)&=
\sum_{i=1}^n \ell_\tau (\Xb_{i}- \Wb_{i})  + \lambda \sum_{i<i'}w_{ii'}\|\Vb_{ii'} \|_2\\
&\qquad + \frac{\rho}{2} \|\Vb  - \Eb\Ub+ \Yb    \|_{\rF}^2+ \frac{\rho}{2} \|\Wb  - \Ub + \Zb    \|_{\rF}^2,
\end{align*}
where $\Wb$, $\Vb$, $\Ub$ are the primal variables, $\Yb$ and $\Zb$ are the dual variables, $\rho$ is a nonnegative tuning parameter for the ADMM algorithm, and $\|\cdot\|_{\rF}$ is the Frobenius norm. 
The updates on both the primal and dual variables can be derived by minimizing the scaled augmented Lagrangian function $L_{\tau}(\Wb,\Vb,\Ub,\Yb,\Zb)$. 

\scolor{We use the ADMM algorithm to solve the above problem over  an increasing sequence of $\lambda$ values until  all data points are clustered into one single cluster.  To enhance computational efficiency, we employ a warm-start strategy, initializing the ADMM algorithm for a new $\lambda$ with the solution obtained from the preceding $\lambda$.
}
Algorithm~\ref{Alg:huberadmm} summarizes routine  for solving \eqref{eq:ccl4}. A detailed derivation of the ADMM updates is deferred to Appendix \ref{app:1}.

\begin{algorithm}[H]
\small
\caption{An alternating direction method of multipliers algorithm.}
\label{Alg:huberadmm}
\begin{enumerate}
\item Input the starting value $\lambda^{(0)}>0$ and the step size $\alpha>1$ of the tuning parameter, resistantification parameter $\tau$, tolerance level $\epsilon$, and $\rho$.
\item  Initialize the primal variables $\Ub^{(0)}$, $\Vb^{(0)} $, $\Wb^{(0)}$, and dual variables $\Yb^{(0)}$ and $\Zb^{(0)}$.
 
\item  Iterate until the stopping criterion:
\begin{itemize}
\item \scolor{ Iterate until convergence:} 
\begin{enumerate}
\item $\Ub^{(t)}=(\Eb^\T\Eb+\Ib)^{-1}[\Eb^{\T}(\Vb^{(t-1)}+\Yb^{(t-1)})+(\Wb^{(t-1)}+\Zb^{(t-1)})]$.
\item For each element in $\Wb^{(t)}$: 
\[
W_{ij}^{(t)} =\begin{cases} \frac{X_{ij}+\rho(U^{(t)}_{ij}-Z^{(t-1)}_{ij})}{1+\rho}, & \mathrm{if} ~ \left|\frac{\rho[X_{ij}-(U^{(t)}_{ij}-Z^{(t-1)}_{ij})]}{1+\rho}\right|\leq \tau,\\
X_{ij}-S(X_{ij}-(U^{(t)}_{ij}-Z^{(t-1)}_{ij}),\tau/{\rho}), & \mathrm{otherwise},
\end{cases}
\]
where $S(a,b) = \text{sign}(a) \max( |a|-b, 0)$ is the soft-thresholding operator.
\item For all $i<i'$, let $\eb_{ii'}^{(t)} = \Ub^{(t)}_{i}-\Ub^{(t)}_{i'}-\Yb^{(t-1)}_{ii'}$ and set
\begin{equation*}
\Vb^{(t)}_{ii'}= \left[1-\frac{\lambda^{(t)} w_{ii'}}{\rho \|\eb_{ii'}^{(t)}\|_2}\right]_+ \eb_{ii'}^{(t)},
\end{equation*}
where $[ a]_+ = \max(0,a)$. 
\item For all $i<i'$, $\Yb_{ii'}^{(t)}= \Yb_{ii'}^{(t-1)}-\rho (\Ub^{(t)}_{i}-\Ub^{(t)}_{i'}-\Vb^{(t)}_{ii'})$.
\item $\Zb^{(t)}=\Zb^{(t-1)}-\rho (\Ub^{(t)}-\Wb^{(t)})$.
\end{enumerate}

\item \scolor{ $\lambda^{(l+1)}=\alpha\lambda^{(l)}$.} 
\end{itemize}
\end{enumerate}
\end{algorithm}

\section{Numerical Studies}
\label{section:numericalstudies}
We examine the performance of the proposed estimator in the presence of arbitrary outliers. We also consider heavy-tailed random noise. In the case of arbitrary outliers, we consider both row-wise contamination models and entry-wise contamination models by randomly contaminating a proportion of rows, or randomly contaminating a proportion of entries in the data matrix. 

We compare the proposed method, abbreviated as proposed, to the least-squares convex clustering~\eqref{eq:ccl}, implemented using the \texttt{R} package \texttt{cvxclustr} \citep{cvxclustr} and thus abbreviated as cvxclustr.
\scolor{In all of our simulation studies, we present the results with specific $\tau$ values, which are provided in the figure captions or tables. 
We start by experimenting with $\tau$ values from a grid with an increment of 0.1, and progressively refine the grid until satisfactory performance is achieved. The results are reported for the best-performing $\tau$ that we identify.
For both methods, we perform multiple iterations for each $\lambda$ and gradually increase $\lambda$ until the optimal solution classifies all data points into a single cluster.
} 
We report the best solution obtained along the path of increasing $\lambda$'s, measured by Hubert and Arabie's adjusted Rand index (HA Rand index) \citep{Rand1971,hubert1985comparing}, \scolor{ whose definition can be found in the appendix.}  A value that is close to one indicates good agreement between the true and estimated clusters.


We implement our proposed method and the least-squares convex clustering method using uniform weight ($w_{ii'} = 1$ for all entries).  
\scolor{We also consider two other competing methods: $k$-means clustering and hierarchical clustering, assuming the true number of clusters is known.}
Additionally, we assume that $n$ observations come from two distinct non-overlapping clusters, $C_1$ and $C_2$, and each cluster contains half of the data points. For different sample sizes, each cluster still contains half of the data.


We generate an $n\times p$ data matrix $\Xb$ according to the model $\Xb_i = \Ub_1 + \bm{\varepsilon}_{i}$ if $i \in C_1$, and $\Xb_i = \Ub_2 + \bm{\varepsilon}_i$, otherwise. The population centroids for the two clusters are constructed as $\Ub_1 \sim \mathcal{N}_p(\mathbf{0},\Ib)$ and $\Ub_2 \sim \mathcal{N}_p((\mathbf{3}_{p/2},\mathbf{-3}_{p/2})^\T,\Ib)$, where $\mathbf{3}_{p/2}$ is a $p/2$-dimensional vector of threes. The random noise vectors $\bvarepsilon_i$'s are independently and identically distributed (i.i.d.) as $\mathcal N({\bf 0},\Ib)$. For row-wise adversarial contamination, we randomly select a proportion of observations (rows of $\Xb$), and in each observation we replace 20\% features with $\mathcal U(10,20)$ random noises. For entry-wise adversarial contamination, we randomly select a certain proportion of entries from the data matrix $\Xb$, and replace them with $\mathcal U(10,20)$ random noises. For heavy-tailed noise, entries of $\bvarepsilon_i$'s are i.i.d. generated according to the centered $t$ distributions with various degrees of freedom and non-centrality parameters to be 0. 
\scolor{
Finally, we examine a different outlier pattern by generating data with Gaussian noise $\mathcal N({\bf 0},\Ib)$ and contamination points drawn from a $t$ distribution (referred to as $t$-outliers), where the degree of freedom is 1 and the non-centrality parameter is 0.
Details on hyper-parameters are provided in Appendix~\ref{hyper}. On a general-purpose laptop with 10 cores and 32GB of memory, our method using uniform weights with $p=20$ and $\tau=0.1$ completes in less than a second for $n=10$, approximately 23 seconds for $n=100$, and 3.5 minutes for $n=200$. Additionally, computation time decreases as $\tau$ decreases. Using Gaussian kernel weights, $w_{ii’}=\exp(-\phi|\mathbf X_i-\mathbf X_{i’}|_2^2)$, with $p=20$, $\tau=0.01$, and $\phi=0.001$, our method takes less than a second for $n=10$, around 26 seconds for $n=100$, and 5 minutes for $n=200$. Increasing $\phi$ results in longer computation times. However, computation time does not significantly increase as $p$ grows.
}

Figures~\ref{entry-wise-1}--\ref{entry-wise-3} summarize the results for entry-wise contamination with those for the row-wise contamination deferred to the appendix. 
\scolor{The values are averaged over 200 replications with different random seeds.
Specifically, Figure~\ref{entry-wise-1}(a) is for $2\%$ entry-wise contamination where we fix the feature dimension $p=20$ but vary the sample size $n$ from 10 to 210 by an increment of 50. 
The left panel displays the HA Rand index, while the right panel presents the estimation error $\|\hat\Ub - \hat\Ub_{\rm clean}\|_{\rF}$, representing the Frobenius norm of the difference between the centroid estimators for contaminated and clean data.
Figure~\ref{entry-wise-1}(b) is for $2\%$ entry-wise contamination where we fix the sample size $n=40$ but vary the feature dimension from 10 to 50 by an increment of 10. 
Figures~\ref{entry-wise-1}(c) and \ref{entry-wise-1}(d) are the same as \ref{entry-wise-1}(a) and \ref{entry-wise-1}(b) except that the entry-wise contamination rate is increased to $10\%$. 
Figures~\ref{entry-wise-2}(a) and \ref{entry-wise-2}(b) illustrate the performance under entry-wise contamination and $t$-noise, where the degree of freedom is 5, for $(n,p) = (20, 10)$ and $(n,p) = (40,20)$, respectively. The contamination proportion is varied from 0.0 to 0.1 in increments of 0.025.
Figures~\ref{entry-wise-2}(c) and (d) are for $t$-noises with varying degrees of freedom, where $n=20,p=10$ and $n=40,p=20$ respectively. 
Figures~\ref{entry-wise-3}(a), \ref{entry-wise-3}(b), and \ref{entry-wise-3}(c) are for Gaussian noises with $t$-outliers  with 1 degree of freedom and varying sample sizes, feature dimensions, and entry-wise contamination proportions, respectively.
In all cases, our proposed method outperforms least-squares convex clustering in terms of the HA Rand index and yields the smallest centroid estimation error among all evaluated methods. We also observe that while the proposed method generally achieves a satisfactory HA Rand index, it is occasionally outperformed by $k$-means. 
This may be due to the fact that $k$-means is implemented with prior knowledge of the true number of clusters. 
For heavy-tailed outliers generated from $t$ distribution with 1 degree of freedom, our proposed method significantly outperforms the other competing methods in terms of estimation error.
In addition,  Figures~\ref{entry-wise-1} and \ref{entry-wise-3} show that the estimation error increases as the sample size $n$ or feature dimension increases. 
This is due to the increasing dimensions of the estimated centroids $\hat\Ub$ and $\hat\Ub_{\rm clean}$, which would generally cause the estimation error $\|\hat\Ub - \hat\Ub_{\rm clean}\|_{\rF}$ to increase. 
Finally, Figures~\ref{entry-wise-2}--\ref{entry-wise-3} indicate that all  methods tend to perform worse, both in terms of the HA Rand index and estimation error, as the level of data contamination increases.
}

To examine the resistance of our proposed method visually, we follow the entry-wise contamination model as in Figure~\ref{entry-wise-1} and generate a data set with two clusters, each with 20 samples.  The feature dimension is picked to be 20 and the entry-wise contamination proportion is set to be $0.075$. 
Observations 0 - 19 belong to cluster 1 and observations 20 - 39 belong to cluster 2.  Figures~\ref{fig:demo}(a) and (b) present the clustering results for least-squares convex clustering and our proposed method while Figures~\ref{fig:demo}(c) and (d) present the dendrograms, both of which are generated following the iterative one-step approximation scheme as in Algorithm \ref{Alg:huberadmm}. Crosses and circles 
indicate the true cluster assignment, while the color indicates the results of clustering methods. The least-squares convex clustering fails to distinguish the two clusters at any $\lambda$ and tends to treat many observations as singleton clusters, while our proposed method (b) manages to almost distinguish the two clusters with only 2 or 3 observations misclustered.

\begin{figure}[!t]
\centering
\subfigure[varying sample sizes, $p=20$, entry-wise contamination = $2\%$, $\tau=1$.]{
\includegraphics[height=0.18\textheight,width=.45\textwidth]{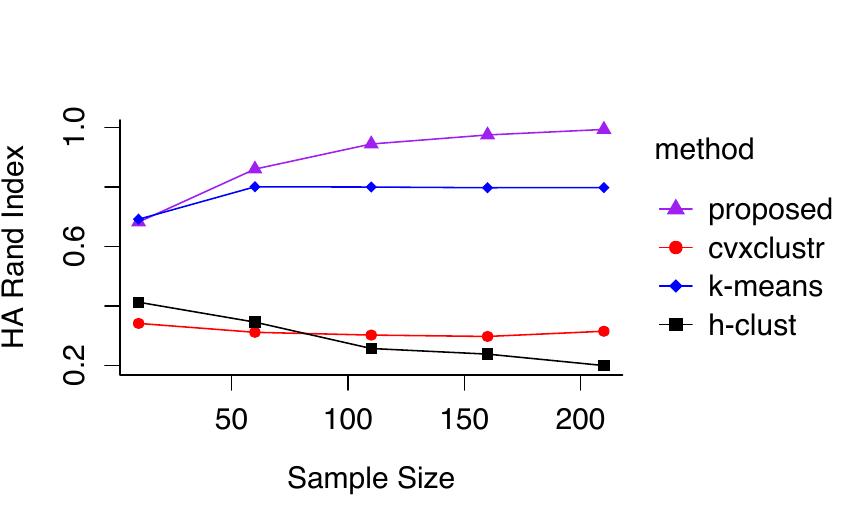}
\hspace{3mm}
\includegraphics[height=0.18\textheight,width=.45\textwidth]{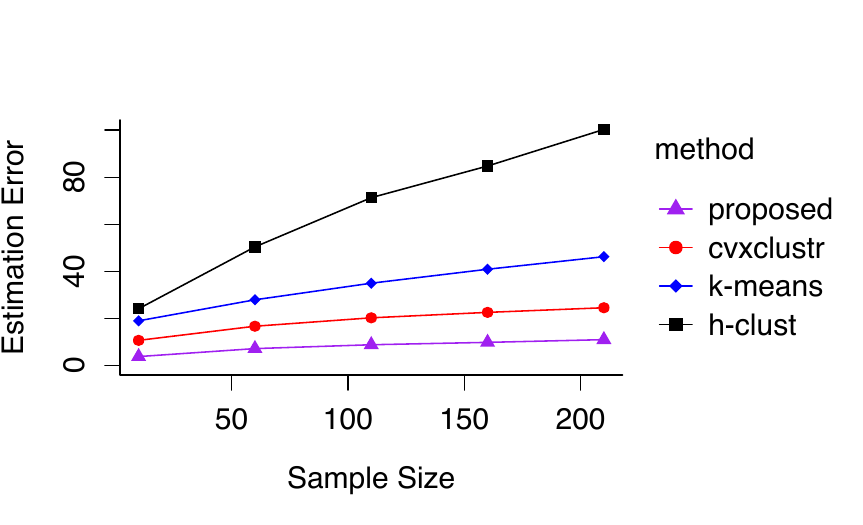}
}
\subfigure[varying feature dimensions, $n = 40$, entry-wise contamination = $2\%$, $\tau=0.1$.]{
\includegraphics[height=0.18\textheight,width=.45\textwidth]{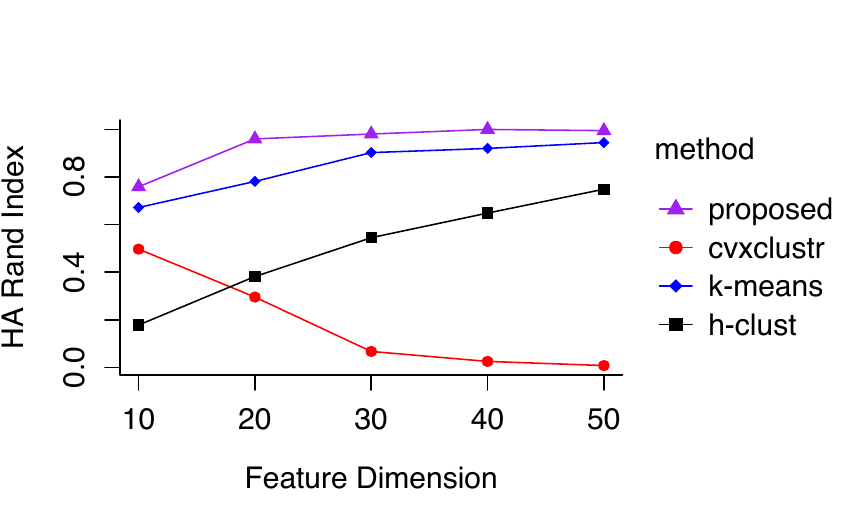}
\hspace{3mm}
\includegraphics[height=0.18\textheight,width=.45\textwidth]{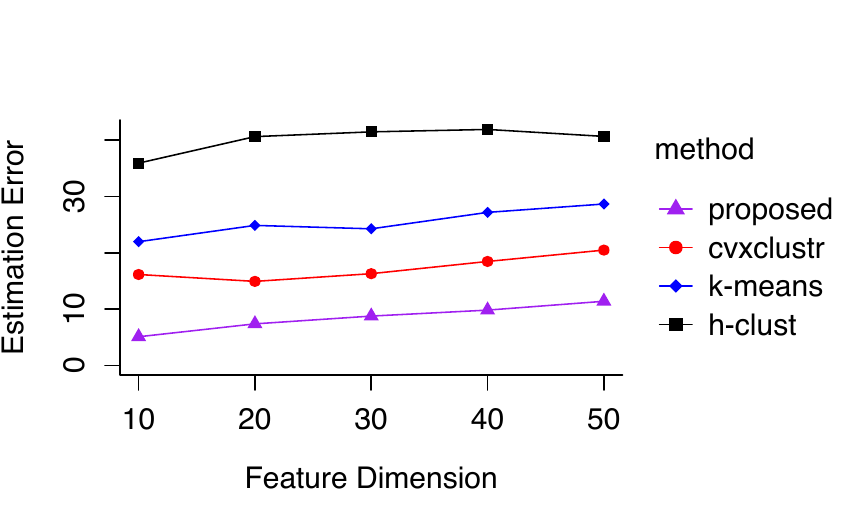}
}
\subfigure[varying sample sizes, $p=20$, entry-wise contamination = $10\%$, $\tau=0.19$.]{
\includegraphics[height=0.18\textheight,width=.45\textwidth]{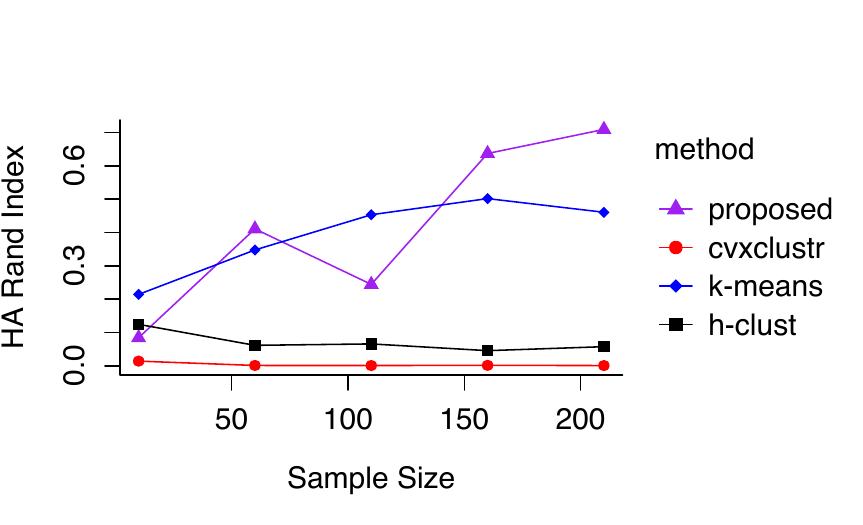}
\hspace{3mm}
\includegraphics[height=0.18\textheight,width=.45\textwidth]{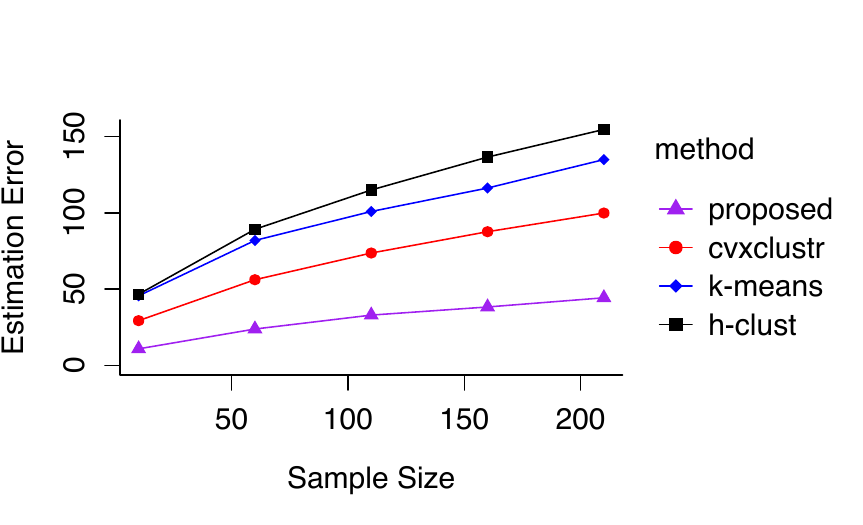}
}
\subfigure[varying feature dimensions, $n = 40$, entry-wise contamination = $10\%$, $\tau=0.08$.]{
\includegraphics[height=0.18\textheight,width=.45\textwidth]{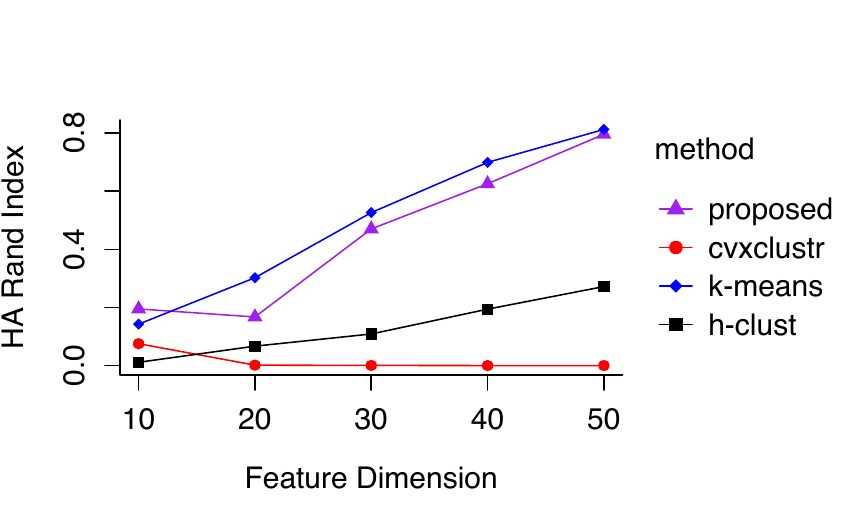}
\hspace{3mm}
\includegraphics[height=0.18\textheight,width=.45\textwidth]{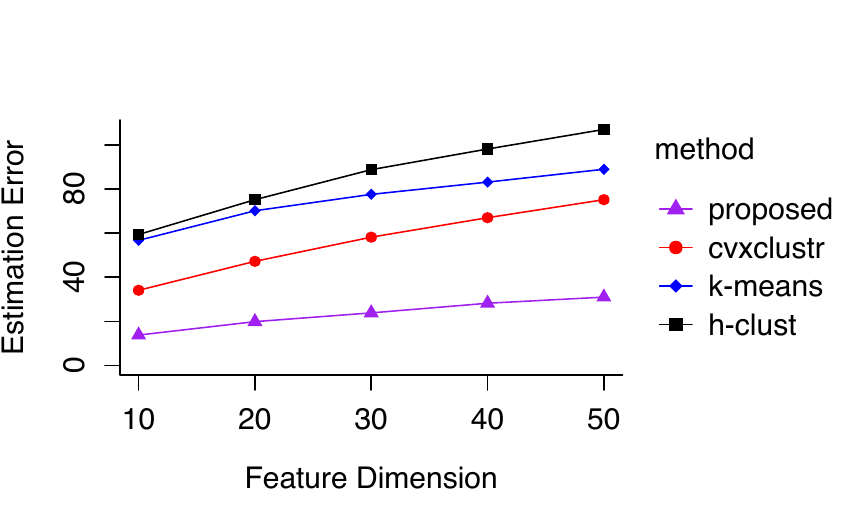}
}
\caption{Comparing our proposed method with others for data with Gaussian noise and uniform outliers with entry-wise contamination. The left panel shows the HA Rand index and the right panel collects the estimation error. 
In all panels, purple, red, blue, and black lines mark our proposed method, least-squares convex clustering, $k$-means, and hierarchical clustering respectively. 
}
\label{entry-wise-1}
\end{figure}

\begin{figure}[!t]
\centering
\subfigure[varying entry-wise outlier proportions, $n = 20, p = 10$, $t$-noise with 5 degrees of freedom, $\tau=0.1$.]{
\includegraphics[height=0.18\textheight,width=.45\textwidth]{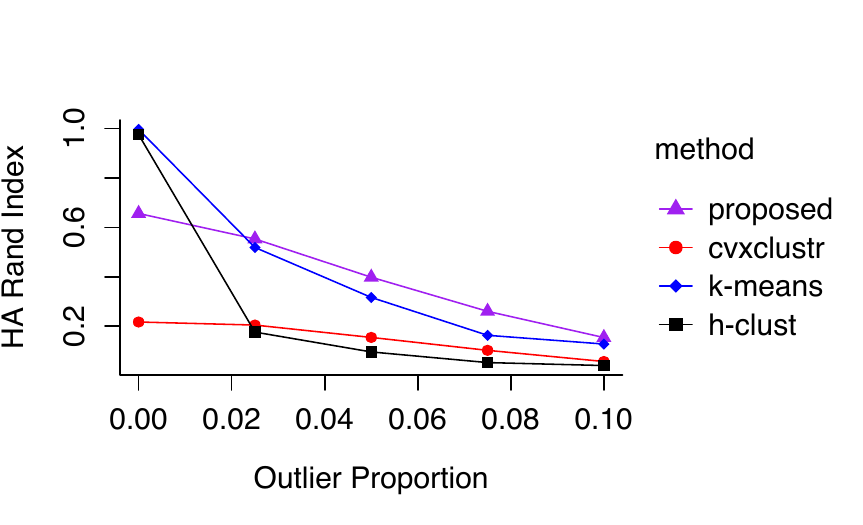}
\hspace{3mm}
\includegraphics[height=0.18\textheight,width=.45\textwidth]{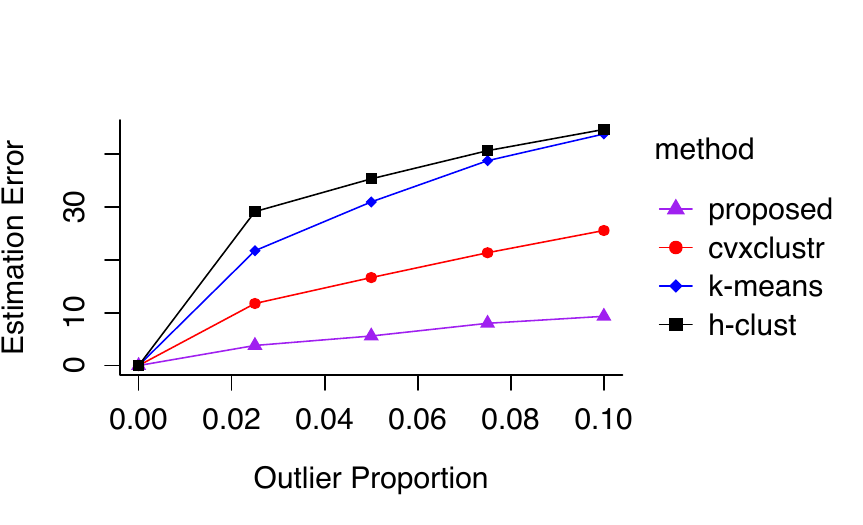}
}
\subfigure[varying entry-wise outlier proportions, $n = 40, p = 20$, $t$-noise with 5 degrees of freedom, $\tau=0.1$.]{
\includegraphics[height=0.18\textheight,width=.45\textwidth]{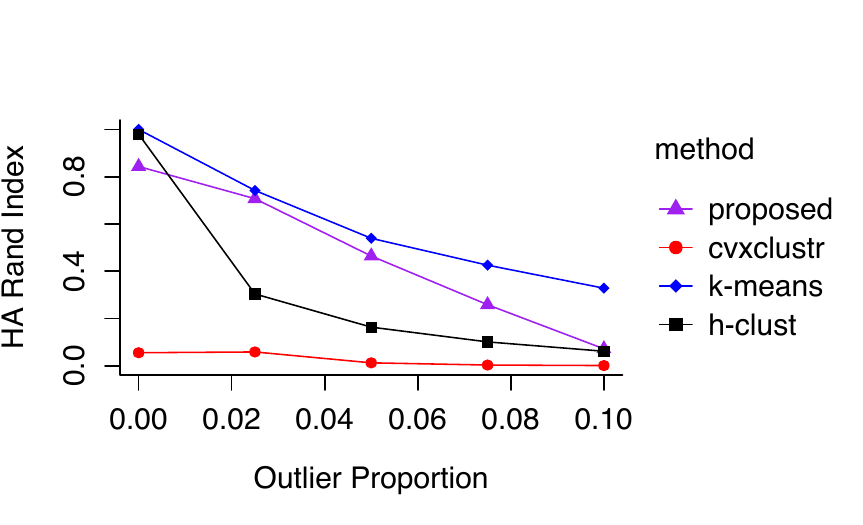}
\hspace{3mm}
\includegraphics[height=0.18\textheight,width=.45\textwidth]{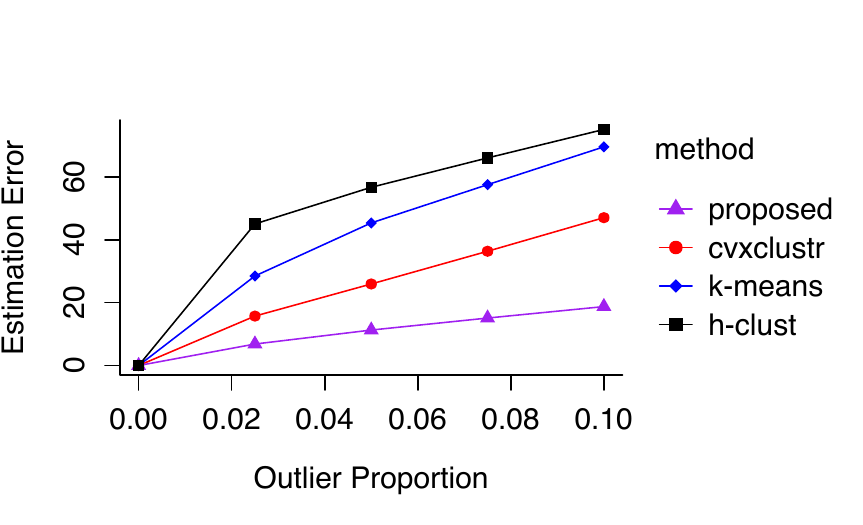}
}
\subfigure[varying degrees of freedom for $t$-noises, $n = 20, p=10$, entry-wise contamination = $2\%$, $\tau=0.1$.]{
\includegraphics[height=0.18\textheight,width=.45\textwidth]{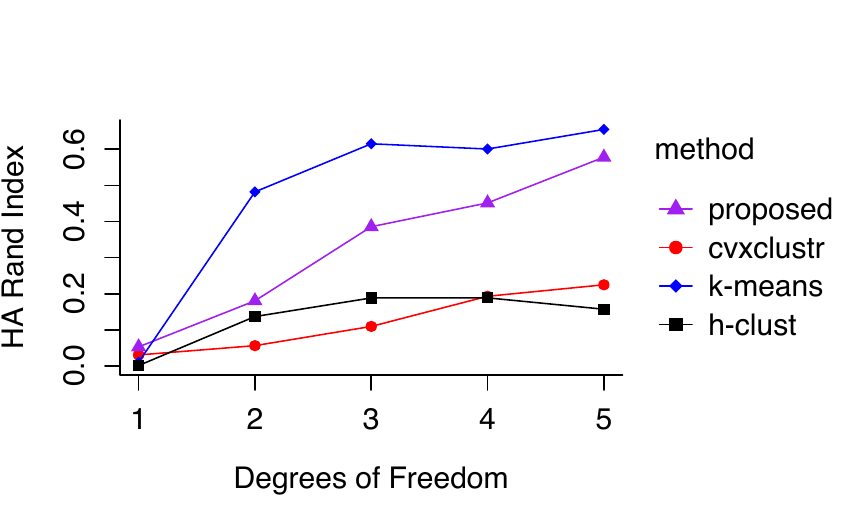}
\hspace{3mm}
\includegraphics[height=0.18\textheight,width=.45\textwidth]{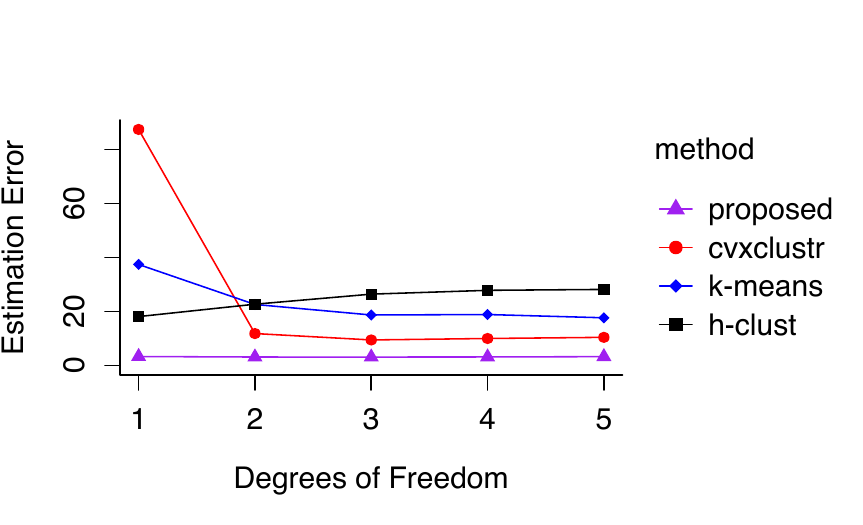}
}
\subfigure[varying degrees of freedom for $t$-noises, $n = 40, p=20$, entry-wise contamination = $2\%$, $\tau=0.1$.]{
\includegraphics[height=0.18\textheight,width=.45\textwidth]{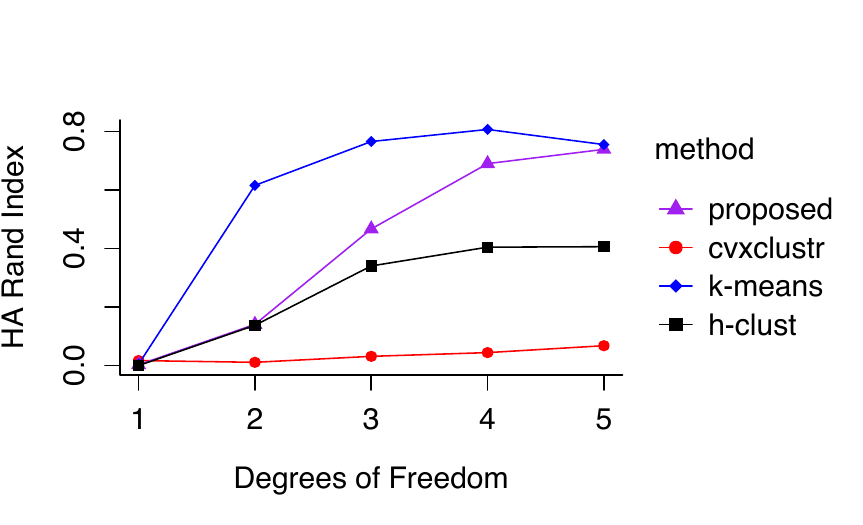}
\hspace{3mm}
\includegraphics[height=0.18\textheight,width=.45\textwidth]{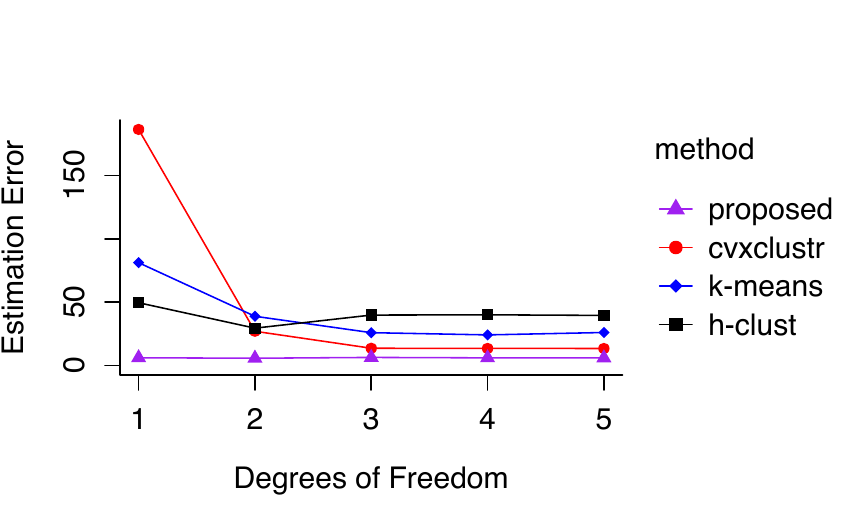}
}
\caption{Comparing our proposed method with others for data with $t$-noise and uniform outliers with entry-wise contamination. 
The left panel shows the HA Rand index and the right panel collects the estimation error. 
In all panels, purple, red, blue, and black lines mark our proposed method, least-squares convex clustering, $k$-means, and hierarchical clustering respectively. 
}
\label{entry-wise-2}
\end{figure}

\begin{figure}[!t]
\centering
\subfigure[varying sample sizes, $p = 20$, entry-wise contamination = $2\%$, $\tau=1$.]{
\includegraphics[height=0.18\textheight,width=.45\textwidth]{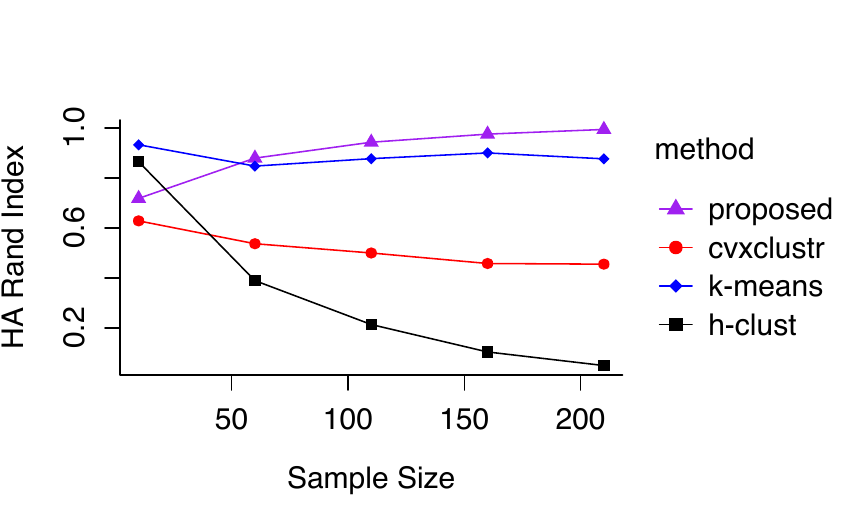}
\hspace{3mm}
\includegraphics[height=0.18\textheight,width=.45\textwidth]{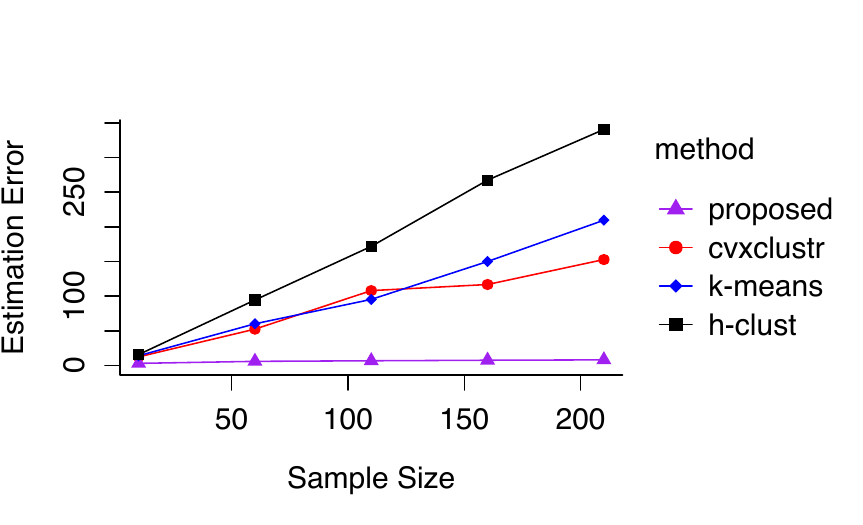}
}
\subfigure[varying feature dimensions, $n = 40$, entry-wise contamination = $2\%$, $\tau=0.1$.]{
\includegraphics[height=0.18\textheight,width=.45\textwidth]{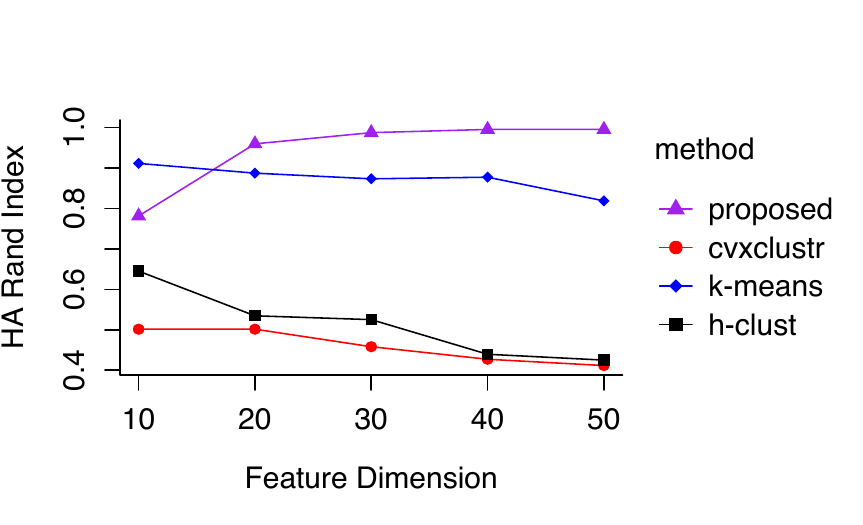}
\hspace{3mm}
\includegraphics[height=0.18\textheight,width=.45\textwidth]{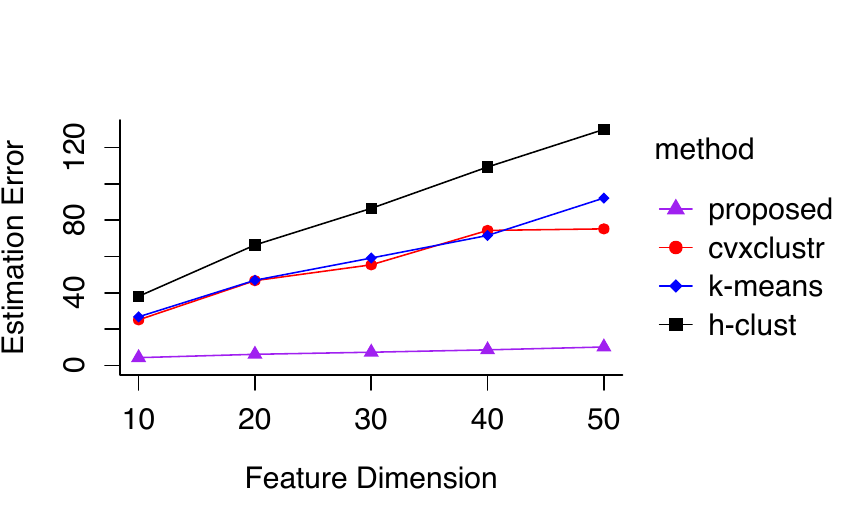}
}
\subfigure[varying entry-wise outlier proportions, $n = 40, p = 20$, $\tau=0.1$.]{
\includegraphics[height=0.18\textheight,width=.45\textwidth]{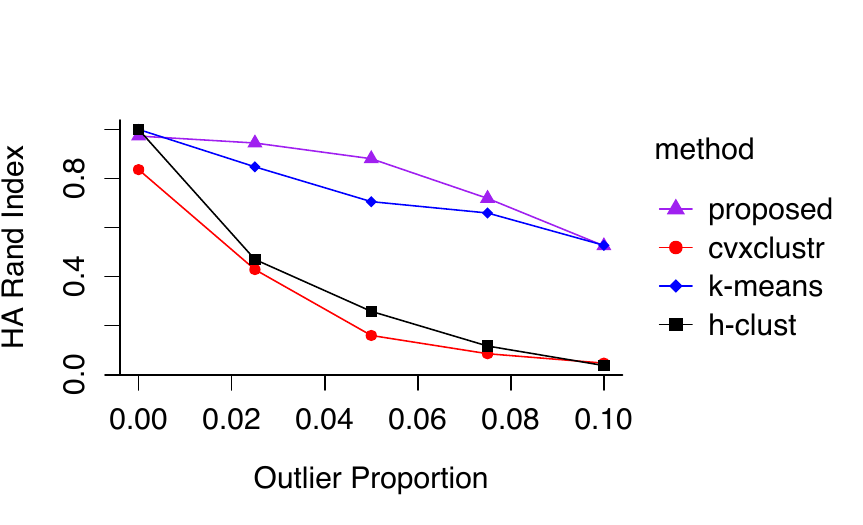}
\hspace{3mm}
\includegraphics[height=0.18\textheight,width=.45\textwidth]{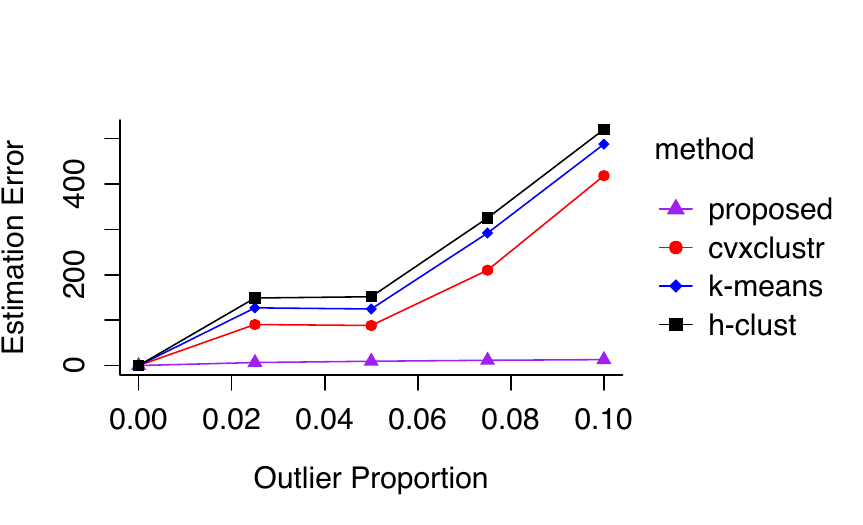}
}
\caption{Comparing our proposed method with others for data with Gaussian noise and $t$-outliers with 1 degree of freedom and entry-wise contamination. The left panel shows the HA Rand index and the right panel collects the estimation error. 
In all panels, purple, red, blue, and black lines mark our proposed method, least-squares convex clustering, $k$-means, and hierarchical clustering respectively. 
}
\label{entry-wise-3}
\end{figure}

\begin{figure}[!t]
\centering
{\includegraphics*[width=\textwidth]{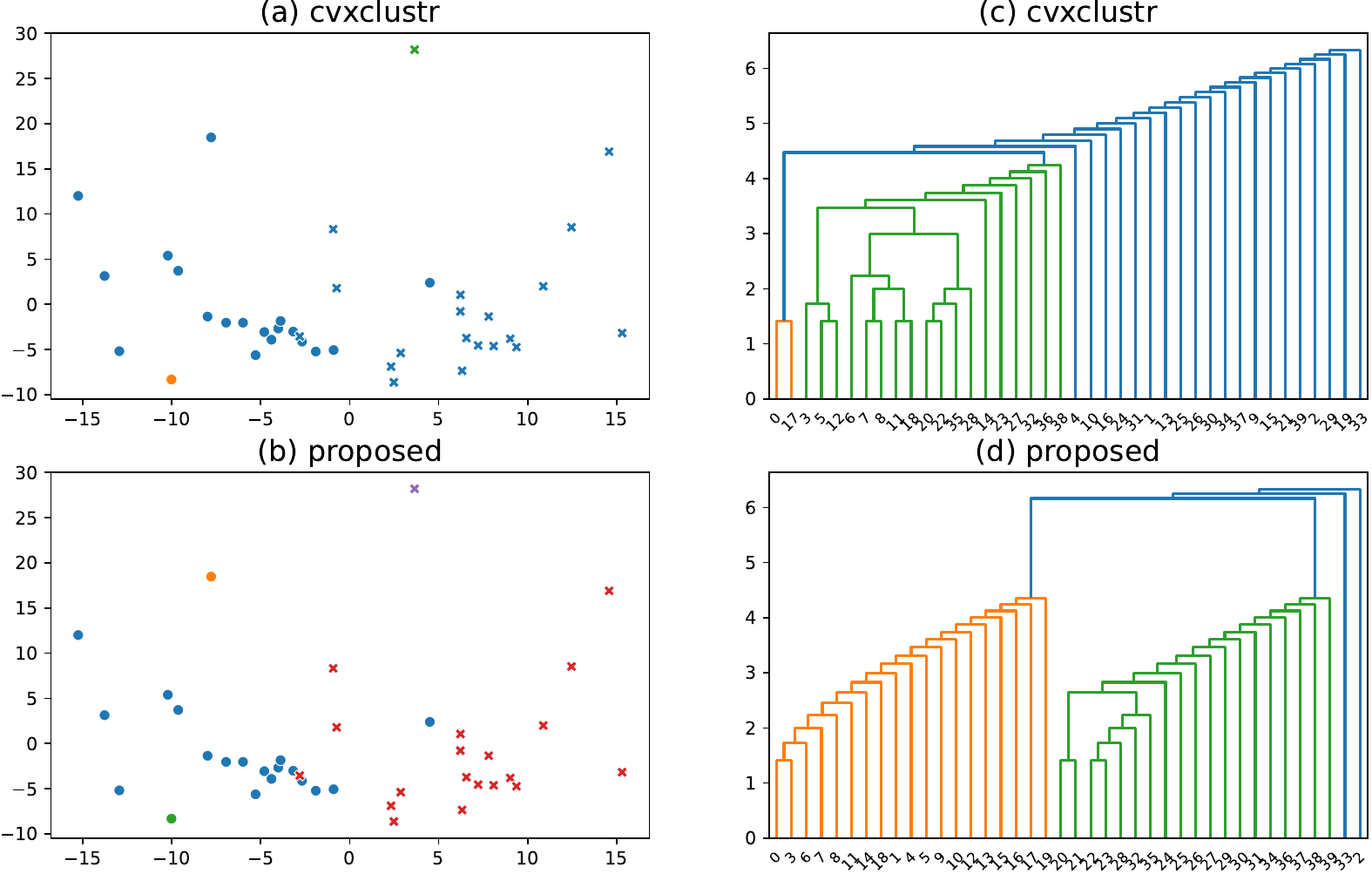}}
\caption{A toy example to demonstrate the resistance. The proposed method identifies the two clusters with high accuracy while the least-squares convex clustering fails to distinguish the two clusters at any $\lambda$.}
\label{fig:demo}
\end{figure}

\scolor{ 
\paragraph{Effect of $\tau$.} 
The resistantification parameter $\tau$ in the Huber loss function, as defined in \eqref{eq:ccl2}, plays an important role in our method.
Intuitively, a smaller $\tau$ enhances robustness, while a larger $\tau$ produces results closer to those of classical convex clustering.
In the previous simulations, we report results for some fixed $\tau$ values, and we now provide some additional simulation results for the effect of $\tau$.
We generate data with Gaussian noise for 1) $(n,p)=(20,10)$ and 2) $(n,p)=(40,20)$, respectively. 
Furthermore, we generate uniform outliers with entry-wise adversarial contamination at proportions of $2\%$ and $10\%$, respectively.
For $2\%$ contamination, we vary $\tau$ from 1 to 21 in increments of 2, while for $10\%$ contamination, we vary $\tau$ from 0.1 to 0.9 in increments of 0.1. 
We set the increment of $\tau$ to 1 for the $2\%$ contamination case to reduce the computational cost, as this is sufficient to reveal the performance pattern as $\tau$ gets larger. 
The results are provided in Figure~\ref{tau-effect}. 
We can see that smaller $\tau$ produces a higher HA Rand index value, and a $\tau$ value of 3 yields the smallest estimation error for $2\%$ contamination. 
In addition, as $\tau$ gets larger, both the HA Rand index and the estimation error of our proposed method become closer to those of the convex clustering method.
}

\scolor{
\paragraph{Overlapping clusters.}
In this paragraph, we examine  the performance of our proposed method when clusters have small inter-cluster distances and overlapping coordinates. 
Specifically, we retain the first cluster centroid as $\Ub_1 \sim \mathcal{N}_p(\mathbf{0},\Ib)$ while modifying the second cluster centroid $\Ub_2$.
For data with Gaussian noise and  entry-wise  contamination with uniform outliers, we generate $\Ub_2 \sim \mathcal{N}_p((\mathbf{0}_{p/2},\mathbf{-3}_{p/2})^\T,\Ib)$; the results are presented in Figure~\ref{overlap-clusters-1}. 
For data with Gaussian noise $\mathcal N({\bf 0},\Ib)$ and  row-wise contamination by  $t$-outliers with 1 degree of freedom, we generated $\Ub_2 \sim \mathcal{N}_p((\mathbf{0}_{p/2},\mathbf{-1}_{p/2})^\T,\Ib)$; the results are presented in Figure~\ref{overlap-clusters-2}. 
For both cases, we implement our proposed method and the least-squares convex clustering method using the Gaussian kernel weights: $w_{ii'}=\exp(-\phi\|\mathbf X_i-\mathbf X_{i'}\|_2^2)$ with parameters $\tau$ and $\phi$ specified in the figure captions. 
We observe that even when the cluster centroids have overlapping coordinates, our proposed method still achieves a satisfactory HA Rand index and significantly outperforms all the other competing methods in minimizing the estimation error.
}

\paragraph{Real-data analysis.}
Finally, we apply our proposed method to several real-world datasets taken from the UCI Machine Learning Repository:  Chemical Composition of Ceramic Samples dataset (Ceramic), Libras Movement dataset (Libras Movement), and Seeds dataset (Seeds). We use cross-validation to select the appropriate $\tau$ for each dataset. 
\scolor{ 
For each dataset, we standardize the features to ensure all variables are on the same scale. As part of the exploratory data analysis, Figure~\ref{real-boxplots} presents boxplots of each feature across the three datasets to highlight key characteristics, particularly potential outliers or data contamination. It is evident that all three datasets, especially the Ceramic and Libras Movement datasets, may contain outliers. 
}

\scolor{We split the datasets into 5, 4, and 4 folds (i.e., equally-sized subsets) for the Seeds, Ceramic and Libras Movement datasets respectively to have an integer number of data points for each fold.} 
For all three datasets, one fold is used for validation while the rest folds are used for training. 
In the validation stage, validation data points are assigned to the clusters formed in the training stage according to the shortest $\ell_2$ distance. 
\scolor{We implement our proposed method and the least-squares convex clustering method using the uniform weights $w_{ii'}=1$, abbreviate as uniform, and the Gaussian kernel weights $w_{ii'}=\exp(-\phi\|\mathbf X_i-\mathbf X_{i'}\|_2^2)$ with $\phi=0.1$, abbreviate as GKernel.}
\scolor{In the training stage, we still apply algorithm~\ref{Alg:huberadmm}}, and the algorithm stops when it reaches the maximal HA rand index for the training data. 
\scolor{For all three data, we pick $\tau$ from 0.001 to 0.01 by an increment of 0.001, and from 0.02 to 1 by an increment of 0.01.} 
We report the average HA rand index in the training stage, and the total average HA Rand index of the training data and validation data together. 
Table~\ref{table1} reports the HA Rand indices and the best $\tau$ selected. For all three datasets considered here, our proposed methods outperform the least-squares convex clustering methods, and our method with Gaussian kernel weights outperforms that with uniform weights.

\begin{table}[!t]
\footnotesize
\begin{center}
\caption{Results for three real-world data sets. The sizes of the data (number of data points in each cluster $\times$ number of clusters) and the HA rand indices are reported.}
\resizebox{\textwidth}{!}{
\scolor{
\begin{tabular}{cccccccc}
  \hline
 \multirow{2}{*}{Method} & \multirow{2}{*}{Weights} &  \multicolumn{2}{c}{Seeds ($70\times 3$)} & \multicolumn{2}{c}{Ceramic ($44\times 2$)} & \multicolumn{2}{c}{Libras Movement ($24\times 15$)} \\
 \cmidrule(lr){3-4} \cmidrule(lr){5-6} \cmidrule(lr){7-8}
& & Train & Total & Train & Total & Train & Total \\
 \hline
\multirow{2}{*}{proposed} 
& uniform & 0.546 ($\tau=0.170$) & 0.493 ($\tau=0.180$) & 0.345 ($\tau=0.060$) & 0.350 ($\tau=0.060$) & 0.147 ($\tau=0.008$) & 0.096 ($\tau=0.680$) \\
 & GKernel & 0.903 ($\tau=0.002$) & 0.719 ($\tau=0.006$) & 0.942 ($\tau=0.001$) & 0.850 ($\tau=0.008$) & 0.365 ($\tau=0.001$) & 0.288 ($\tau=0.790$) \\
\hline
\multirow{2}{*}{cvxclustr} 
& uniform & 0.037 & 0.247 & 0.004 & 0.143 & 0.014 & 0.056 \\ 
& GKernel & 0.324 & 0.408 & 0.665 & 0.493 & 0.318 & 0.266 \\
\hline
\end{tabular}
}
}
\label{table1}
\end{center}
\end{table}

\section{Conclusions}
We propose a resistant convex clustering method and a corresponding ADMM algorithm. 
Theoretically, we analyze the breakdown point of the proposed resistant convex clustering
method. We show that the proposed estimator does not break down until more than half of the observations are
arbitrary outliers. This is somewhat surprising, at least to us, as we expected
one arbitrary large outlier will destroy the clustering procedure because there are as
many parameters as the samples. Compared with the estimator without the fusion
penalty, we find that the fusion penalty helps enhance the resistance of the clustering procedure by fusing the estimators of the centroids to the cluster centroids of uncontaminated observations,  but not the other way around. Indeed, the least-squares convex clustering with the fusion penalty breaks down when there is only one adversarial sample, where the fusion penalty does not help enhance the resistance property for the least-squares loss. This demonstrates the necessity of using the Huber loss function.  We conjecture the phenomenon of enhancing resistance/robustness would hold for general graph-type penalties/constraints. 
As future work, our proposed method can be extended to biclustering problem \citep{chi2014convex} and co-clustering problem for tensors \citep{chi2018provable}.

\bibliography{reference}

\begin{thebibliography}{}

\bibitem[Avella-Medina et~al., 2018]{avella2018robust}
Avella-Medina, M., Battey, H.~S., Fan, J., and Li, Q. (2018).
\newblock Robust estimation of high-dimensional covariance and precision
  matrices.
\newblock {\em Biometrika}, 105(2):271--284.

\bibitem[Boyd and Vandenberghe, 2004]{boyd2004convex}
Boyd, S. and Vandenberghe, L. (2004).
\newblock {\em Convex Optimization}.
\newblock Cambridge University Press, New York.

\bibitem[Catoni, 2012]{catoni2012challenging}
Catoni, O. (2012).
\newblock Challenging the empirical mean and empirical variance: A deviation
  study.
\newblock {\em Annales de I'Institut Henri Poincar{\'e} - Probabilit{\'e}s et
  Statistiques}, 48:1148--1185.

\bibitem[Chen et~al., 2015]{chen2014convex}
Chen, G.~K., Chi, E.~C., Ranola, J., and Lange, K. (2015).
\newblock Convex clustering: An attractive alternative to hierarchical
  clustering.
\newblock {\em PLOS Computational Biology}, 11(5):e1004228.

\bibitem[Chi et~al., 2017]{chi2014convex}
Chi, E.~C., Allen, G.~I., and Baraniuk, R.~G. (2017).
\newblock Convex biclustering.
\newblock {\em Biometrics}, 73(1):10--19.

\bibitem[Chi et~al., 2018]{chi2018provable}
Chi, E.~C., Gaines, B.~R., Sun, W.~W., Zhou, H., and Yang, J. (2018).
\newblock Provable convex co-clustering of tensors.
\newblock \emph{arXiv preprint arXiv:1803.06518}.

\bibitem[Chi and Lange, 2014]{cvxclustr}
Chi, E.~C. and Lange, K. (2014).
\newblock {\em cvxclustr: Splitting methods for convex clustering}.
\newblock URL {http://cran.r-project.org/web/packages/cvxclustr.} {R} package
  version 1.1.1.

\bibitem[Chi and Lange, 2015]{chi2013splitting}
Chi, E.~C. and Lange, K. (2015).
\newblock Splitting methods for convex clustering.
\newblock {\em Journal of Computational and Graphical Statistics},
  24(4):994--1013.

\bibitem[Chi and Steinerberger, 2018]{chi2018recovering}
Chi, E.~C. and Steinerberger, S. (2018).
\newblock Recovering trees with convex clustering.
\newblock \emph{arXiv preprint arXiv:1806.11096}.

\bibitem[Cuesta-Albertos et~al., 2008]{cuesta2008robust}
Cuesta-Albertos, J., Matr{\'a}n, C., and Mayo-Iscar, A. (2008).
\newblock Robust estimation in the normal mixture model based on robust
  clustering.
\newblock {\em Journal of the Royal Statistical Society Series B: Statistical
  Methodology}, 70(4):779--802.

\bibitem[Cuesta-Albertos et~al., 1997]{cuesta1997trimmed}
Cuesta-Albertos, J.~A., Gordaliza, A., and Matr{\'a}n, C. (1997).
\newblock Trimmed $ k $-means: an attempt to robustify quantizers.
\newblock {\em The Annals of Statistics}, 25(2):553--576.

\bibitem[Donoho and Huber, 1983]{donoho1983notion}
Donoho, D.~L. and Huber, P.~J. (1983).
\newblock The notion of breakdown point.
\newblock In {\em A Festschrift For Erich L. Lehmann}, pages 157--184. Belmont,
  Wadsworth.

\bibitem[Dorabiala et~al., 2022]{dorabiala2022robust}
Dorabiala, O., Kutz, J.~N., and Aravkin, A.~Y. (2022).
\newblock Robust trimmed k-means.
\newblock {\em Pattern Recognition Letters}, 161:9--16.

\bibitem[Gallegos and Ritter, 2005]{gallegos2005robust}
Gallegos, M.~T. and Ritter, G. (2005).
\newblock A robust method for cluster analysis.
\newblock {\em Annals of Statistics}, pages 347--380.

\bibitem[Garcia-Escudero and Gordaliza, 1999]{garcia1999robustness}
Garcia-Escudero, L.~A. and Gordaliza, A. (1999).
\newblock Robustness properties of k means and trimmed k means.
\newblock {\em Journal of the American Statistical Association},
  94(447):956--969.

\bibitem[Garc{\'\i}a-Escudero et~al., 2010]{garcia2010review}
Garc{\'\i}a-Escudero, L.~A., Gordaliza, A., Matr{\'a}n, C., and Mayo-Iscar, A.
  (2010).
\newblock A review of robust clustering methods.
\newblock {\em Advances in Data Analysis and Classification}, 4:89--109.

\bibitem[Georgogiannis, 2016]{georgogiannis2016robust}
Georgogiannis, A. (2016).
\newblock Robust k-means: a theoretical revisit.
\newblock {\em Advances in Neural Information Processing Systems}, 29.

\bibitem[Hampel, 1971]{hampel1971general}
Hampel, F.~R. (1971).
\newblock A general qualitative definition of robustness.
\newblock {\em The Annals of Mathematical Statistics}, 42(6):1887--1896.

\bibitem[Hastie et~al., 2009]{ElemStatLearn}
Hastie, T., Tibshirani, R., and Friedman, J. (2009).
\newblock {\em The Elements of Statistical Learning: Data Mining, Inference and
  Prediction}.
\newblock Springer, New York.

\bibitem[Hocking et~al., 2011]{hocking2011clusterpath}
Hocking, T.~D., Joulin, A., Bach, F., and Vert, J.-P. (2011).
\newblock Clusterpath: An algorithm for clustering using convex fusion
  penalties.
\newblock In {\em Proceedings of the 28th International Conference on Machine
  Learning}.

\bibitem[Huber, 1964]{huber1964}
Huber, P.~J. (1964).
\newblock Robust estimation of a location parameter.
\newblock {\em The Annals of Mathematical Statistics}, 35(1):73--101.

\bibitem[Huber, 1973]{huber1973}
Huber, P.~J. (1973).
\newblock Robust regression: Asymptotics, conjectures and monte carlo.
\newblock {\em The Annals of Statistics.}, 1(5):799--821.

\bibitem[Hubert and Arabie, 1985]{hubert1985comparing}
Hubert, L. and Arabie, P. (1985).
\newblock {Comparing partitions}.
\newblock {\em Journal of classification}, 2(1):193--218.

\bibitem[Johnson and Wichern, 2002]{johnson2002applied}
Johnson, R.~A. and Wichern, D.~W. (2002).
\newblock {\em Applied Multivariate Statistical Analysis}.
\newblock Prentice Hall, New Jersey.

\bibitem[Ke et~al., 2019]{ke2018user}
Ke, Y., Minsker, S., Ren, Z., Sun, Q., and Zhou, W.-X. (2019).
\newblock User-friendly covariance estimation for heavy-tailed distributions.
\newblock \emph{Statistical Science, in press}.

\bibitem[Lin et~al., 2007]{lin2007robust}
Lin, T.~I., Lee, J.~C., and Hsieh, W.~J. (2007).
\newblock Robust mixture modeling using the skew $t$ distribution.
\newblock {\em Statistics and computing}, 17:81--92.

\bibitem[Lindsten et~al., 2011]{lindsten2011clustering}
Lindsten, F., Ohlsson, H., and Ljung, L. (2011).
\newblock Clustering using sum-of-norms regularization: With application to
  particle filter output computation.
\newblock In {\em 2011 IEEE Statistical Signal Processing Workshop (SSP)},
  pages 201--204.

\bibitem[McLachlan et~al., 2019]{mclachlan2019finite}
McLachlan, G.~J., Lee, S.~X., and Rathnayake, S.~I. (2019).
\newblock Finite mixture models.
\newblock {\em Annual review of statistics and its application}, 6:355--378.

\bibitem[Mizera and M{\"u}ller, 1999]{mizera1999breakdown}
Mizera, I. and M{\"u}ller, C.~H. (1999).
\newblock Breakdown points and variation exponents of robust $ m $-estimators
  in linear models.
\newblock {\em The Annals of Statistics}, 27(4):1164--1177.

\bibitem[Peel and McLachlan, 2000]{peel2000robust}
Peel, D. and McLachlan, G.~J. (2000).
\newblock Robust mixture modelling using the t distribution.
\newblock {\em Statistics and computing}, 10:339--348.

\bibitem[Pelckmans et~al., 2005]{pelckmans2005convex}
Pelckmans, K., De~Brabanter, J., Suykens, J., and De~Moor, B. (2005).
\newblock Convex clustering shrinkage.
\newblock In {\em PASCAL Workshop on Statistics and Optimization of Clustering
  Workshop}.

\bibitem[Radchenko and Mukherjee, 2017]{radchenko2014consistent}
Radchenko, P. and Mukherjee, G. (2017).
\newblock Consistent clustering via $\ell_1$ fusion penalty.
\newblock {\em Journal of the Royal Statistical Society: Series B (Statistical
  Methodology)}, 79(5):1527--1546.

\bibitem[Rand, 1971]{Rand1971}
Rand, W.~M. (1971).
\newblock Objective criteria for the evaluation of clustering methods.
\newblock {\em Journal of the American Statistical Association},
  66(336):846--850.

\bibitem[Rousseeuw, 1984]{rousseeuw1984least}
Rousseeuw, P.~J. (1984).
\newblock Least median of squares regression.
\newblock {\em Journal of the American Statistical Association},
  79(388):871--880.

\bibitem[Rousseeuw and Yohai, 1984]{rousseeuw1984robust}
Rousseeuw, P.~J. and Yohai, V. (1984).
\newblock Robust regression by means of {S}-estimators.
\newblock In {\em Robust and Nonlinear Time Series Analysis}, pages 256--272.
  Springer.

\bibitem[Salibian-Barrera and Zamar, 2002]{salibian2002bootrapping}
Salibian-Barrera, M. and Zamar, R.~H. (2002).
\newblock Bootrapping robust estimates of regression.
\newblock {\em The Annals of Statistics}, 30(2):556--582.

\bibitem[Sun et~al., 2018]{sun2018convex}
Sun, D., Toh, K.-C., and Yuan, Y. (2018).
\newblock Convex clustering: Model, theoretical guarantee and efficient
  algorithm.
\newblock \emph{arXiv preprint arXiv:1810.02677}.

\bibitem[Sun et~al., 2020]{sun2018adaptive}
Sun, Q., Zhou, W.-X., and Fan, J. (2020).
\newblock Adaptive huber regression.
\newblock {\em Journal of the American Statistical Association},
  115(529):254--265.

\bibitem[Tan et~al., 2018]{tan2018robust}
Tan, K.~M., Sun, Q., and Witten, D.~M. (2018).
\newblock Robust sparse reduced rank regression in high dimensions.
\newblock \emph{arXiv preprint arXiv:1810.07913v2}.

\bibitem[Tan and Witten, 2015]{tan2015}
Tan, K.~M. and Witten, D.~M. (2015).
\newblock {Statistical properties of convex clustering}.
\newblock {\em Electronic Journal of Statistics}, 9:2324--2347.

\bibitem[Wang et~al., 2018]{wang2016sparse}
Wang, B., Zhang, Y., Sun, W.~W., and Fang, Y. (2018).
\newblock Sparse convex clustering.
\newblock {\em Journal of Computational and Graphical Statistics},
  27(2):393--403.

\bibitem[Weylandt et~al., 2020]{Weylandt2020dynamic}
Weylandt, M., Nagorski, J., and Allen, G.~I. (2020).
\newblock Dynamic visualization and fast computation for convex clustering via
  algorithmic regularization.
\newblock {\em Journal of Computational and Graphical Statistics},
  29(1):87--96.

\bibitem[Whang et~al., 2015]{whang2015non}
Whang, J.~J., Dhillon, I.~S., and Gleich, D.~F. (2015).
\newblock Non-exhaustive, overlapping k-means.
\newblock In {\em Proceedings of the 2015 SIAM international conference on data
  mining}, pages 936--944. SIAM.

\bibitem[Yang et~al., 2012]{yang2012robust}
Yang, M.-S., Lai, C.-Y., and Lin, C.-Y. (2012).
\newblock A robust em clustering algorithm for gaussian mixture models.
\newblock {\em Pattern Recognition}, 45(11):3950--3961.

\bibitem[Yohai, 1987]{yohai1987high}
Yohai, V.~J. (1987).
\newblock High breakdown-point and high efficiency robust estimates for
  regression.
\newblock {\em The Annals of Statistics}, 15:642--656.

\bibitem[Zhu et~al., 2014]{NIPS2014_5307}
Zhu, C., Xu, H., Leng, C., and Yan, S. (2014).
\newblock Convex optimization procedure for clustering: {T}heoretical revisit.
\newblock In {\em Advances in Neural Information Processing Systems 27}.

\end{thebibliography}
\bibliographystyle{apalike}




\begin{figure}[!t]
\centering
\subfigure[$\tau$ ranges from 1 to 21 by an increment of 2, $n=20, p=10$, entry-wise contamination = $2\%$.]{
\includegraphics[height=0.18\textheight,width=.45\textwidth]{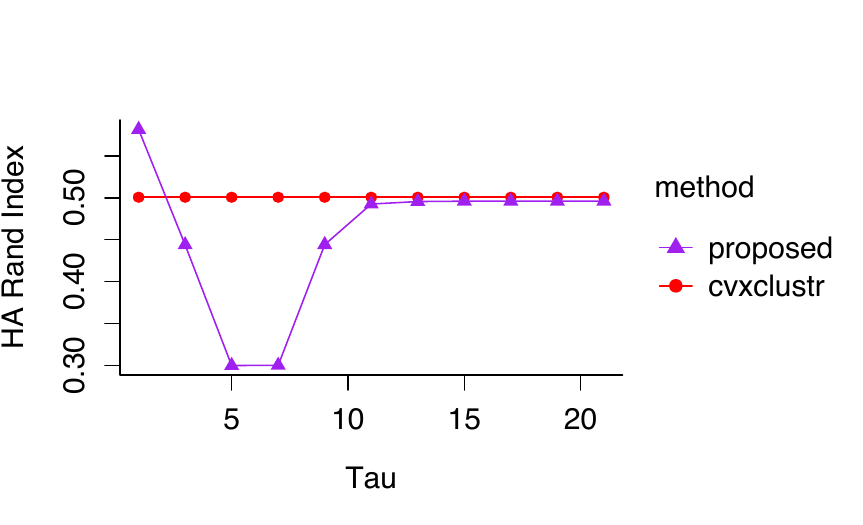}
\hspace{3mm}
\includegraphics[height=0.18\textheight,width=.45\textwidth]{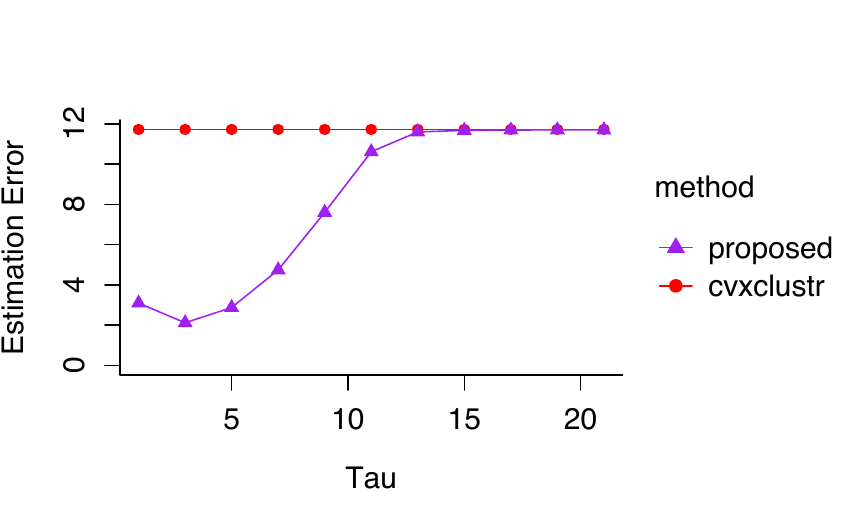}
}
\subfigure[$\tau$ ranges from 1 to 21 by an increment of 2, $n=40, p=20$, entry-wise contamination = $2\%$.]{
\includegraphics[height=0.18\textheight,width=.45\textwidth]{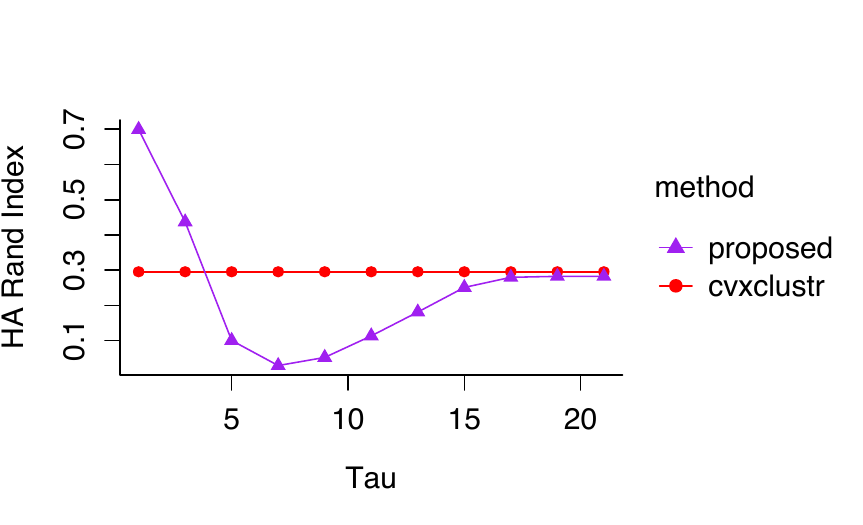}
\hspace{3mm}
\includegraphics[height=0.18\textheight,width=.45\textwidth]{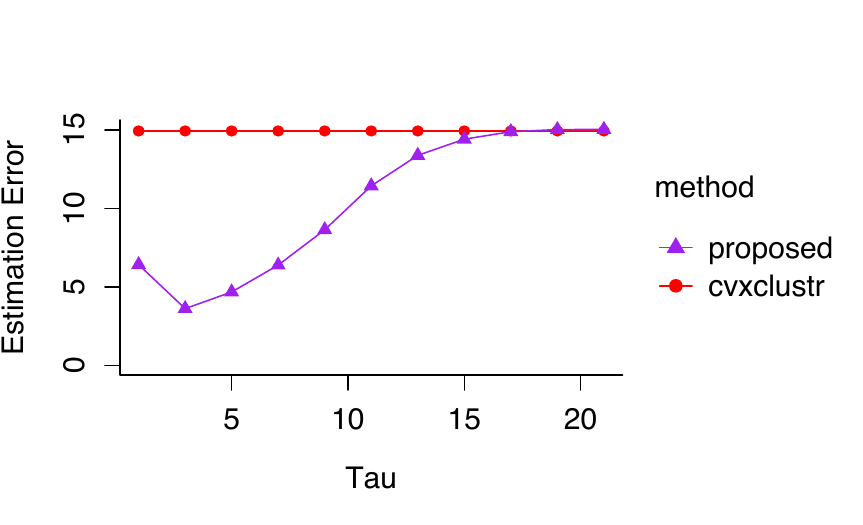}
}
\subfigure[$\tau$ ranges from 0.1 to 0.9 by an increment of 0.1, $n=20, p=10$, entry-wise contamination = $10\%$.]{
\includegraphics[height=0.18\textheight,width=.45\textwidth]{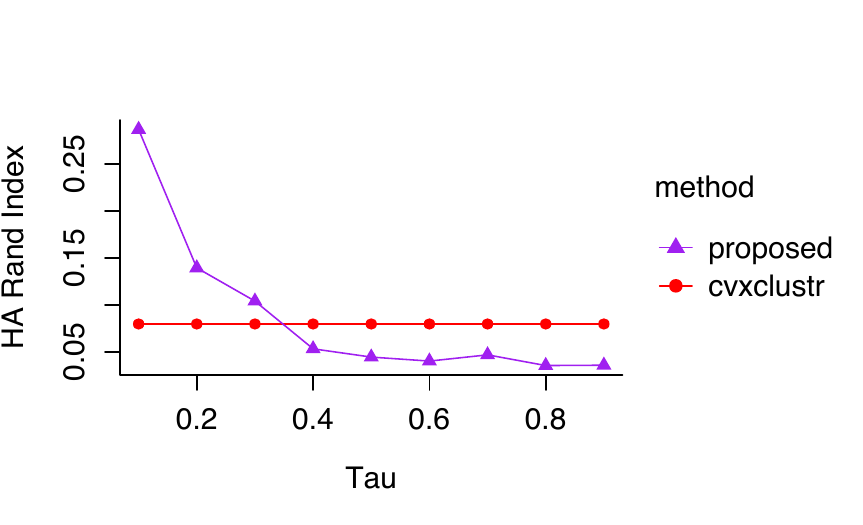}
\hspace{3mm}
\includegraphics[height=0.18\textheight,width=.45\textwidth]{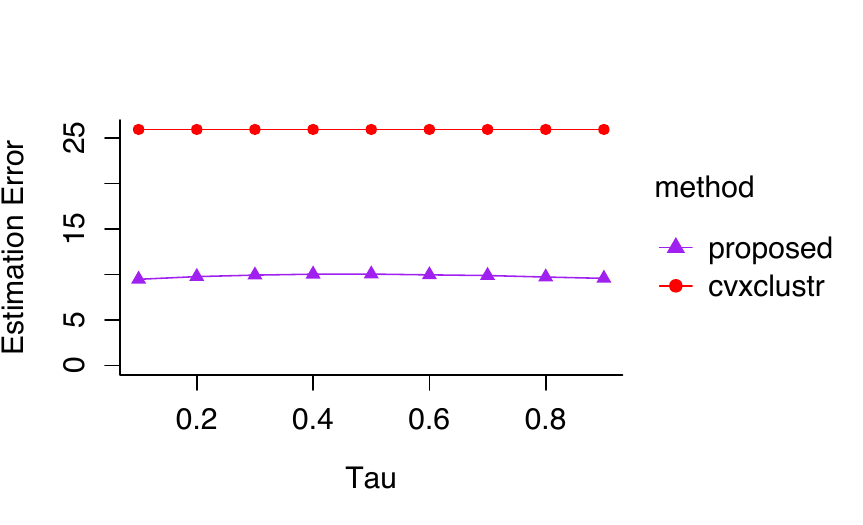}
}
\subfigure[$\tau$ ranges from 0.1 to 0.9 by an increment of 0.1, $n=40, p=20$, entry-wise contamination = $10\%$.]{
\includegraphics[height=0.18\textheight,width=.45\textwidth]{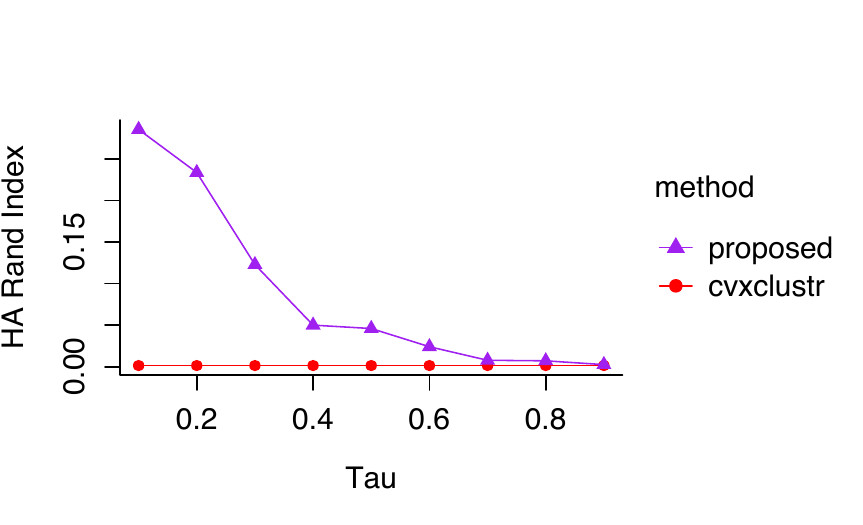}
\hspace{3mm}
\includegraphics[height=0.18\textheight,width=.45\textwidth]{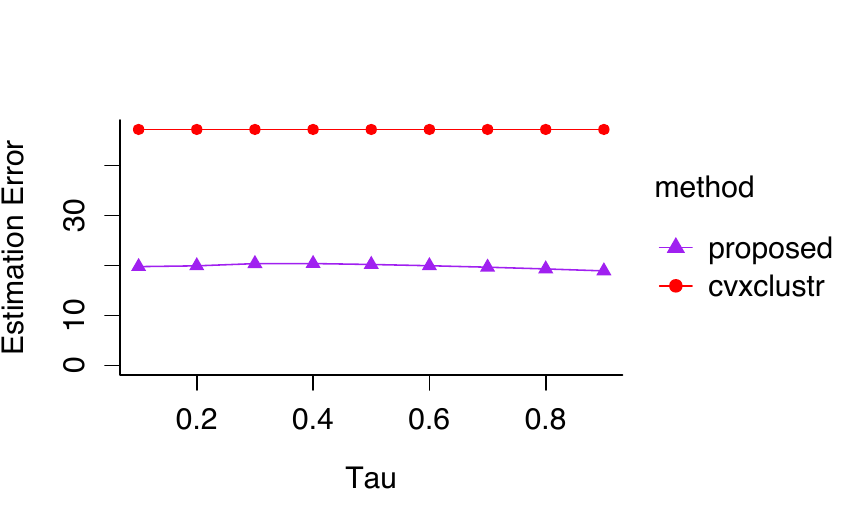}
}
\caption{Investigating the effect of $\tau$ for data with Gaussian noise and uniform outliers with entry-wise contamination. The left panel shows the HA Rand index and the right panel collects the estimation error. 
In all panels, purple and red lines mark our proposed method and least-squares convex clustering respectively. 
}
\label{tau-effect}
\end{figure}

\begin{figure}[!t]
\centering
\subfigure[varying sample sizes, $p = 20$, entry-wise contamination = $2\%$, $\tau=0.01, \phi=0.001$.]{
\includegraphics[height=0.18\textheight,width=.45\textwidth]{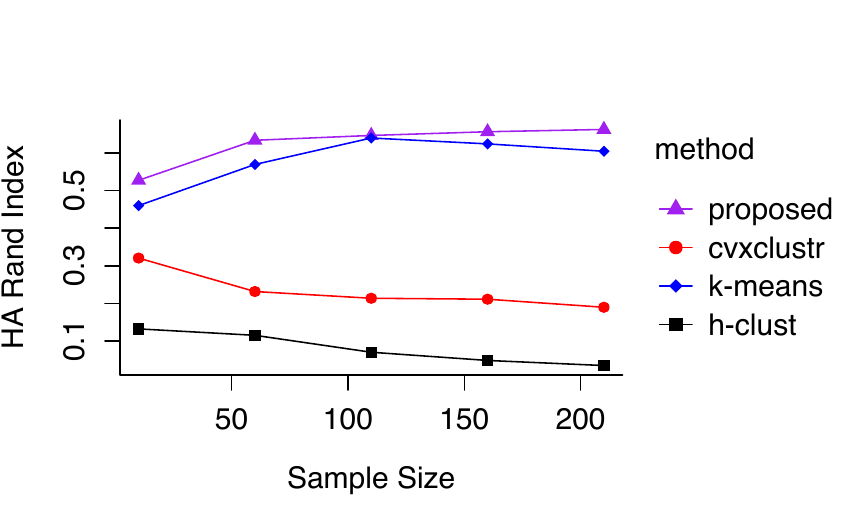}
\hspace{3mm}
\includegraphics[height=0.18\textheight,width=.45\textwidth]{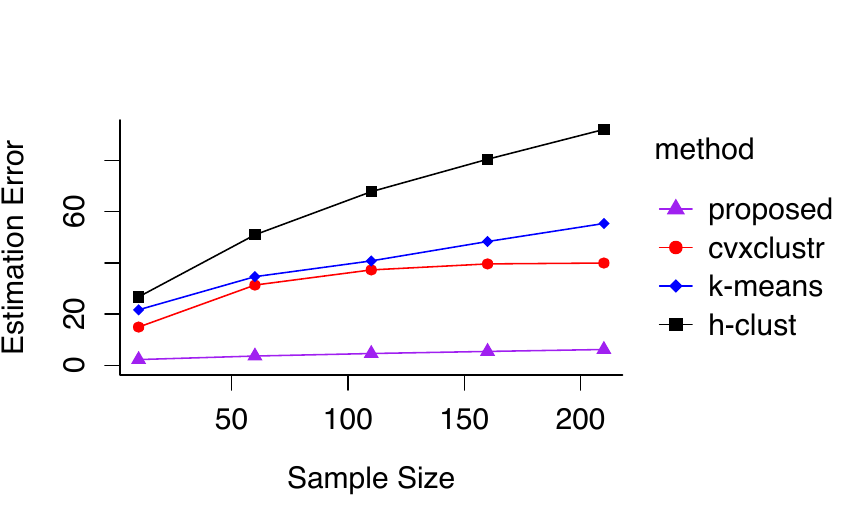}
}
\subfigure[varying feature dimensions, $n = 40$, entry-wise contamination = $2\%$, $\tau=0.01, \phi=0.001$.]{
\includegraphics[height=0.18\textheight,width=.45\textwidth]{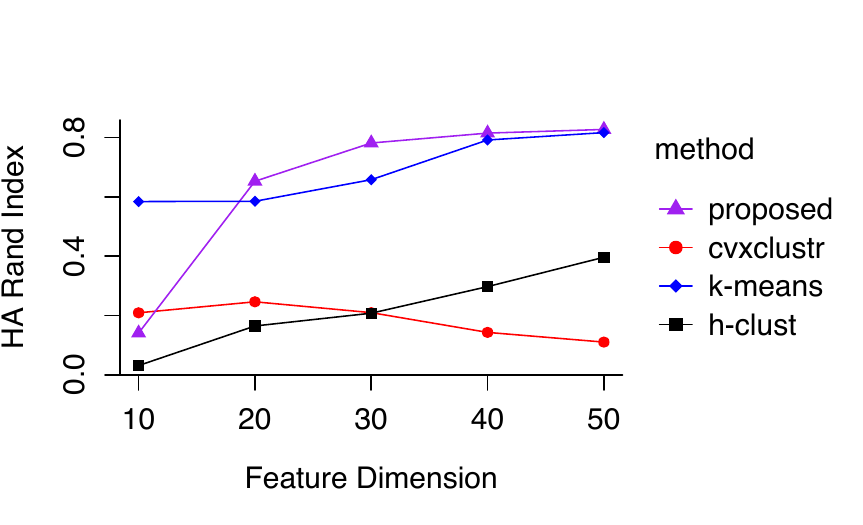}
\hspace{3mm}
\includegraphics[height=0.18\textheight,width=.45\textwidth]{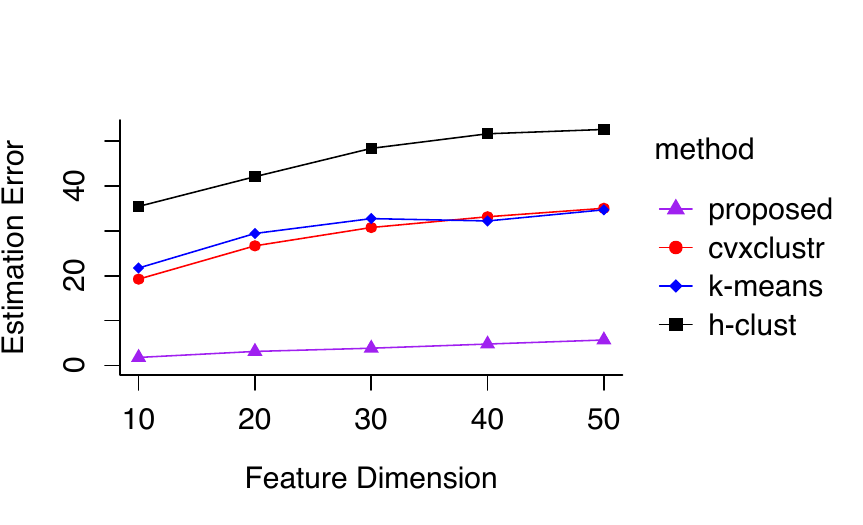}
}
\subfigure[varying entry-wise outlier proportions, $n = 40, p = 20$, $\tau=0.01, \phi=0.002$.]{
\includegraphics[height=0.18\textheight,width=.45\textwidth]{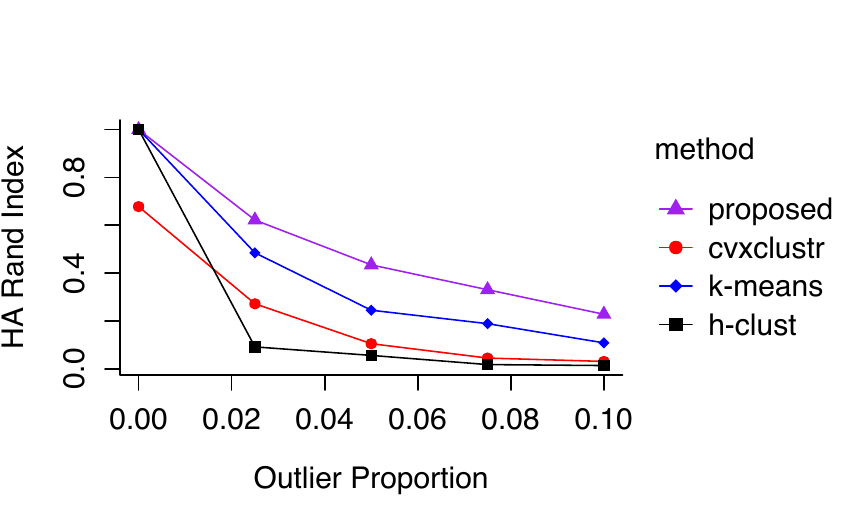}
\hspace{3mm}
\includegraphics[height=0.18\textheight,width=.45\textwidth]{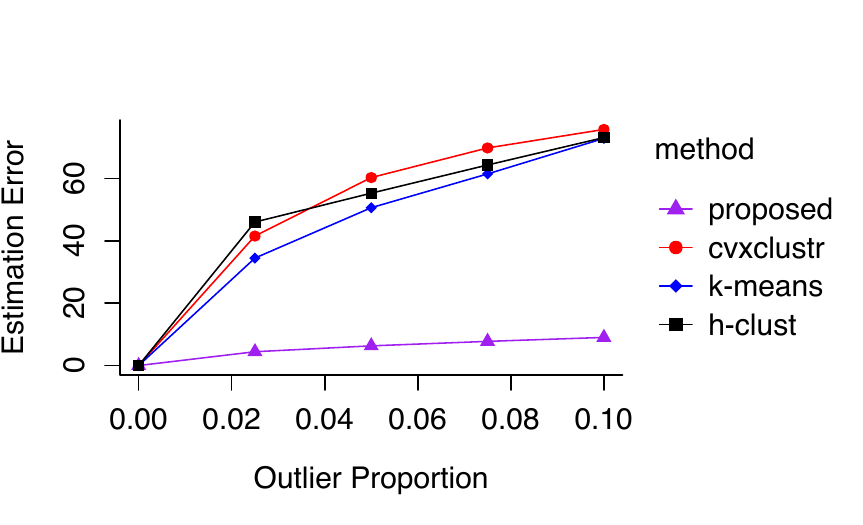}
}
\caption{Investigating the effect of overlapping clusters. 
Data generated from clusters $\Ub_1 \sim \mathcal{N}_p(\mathbf{0},\Ib)$ and $\Ub_2 \sim \mathcal{N}_p((\mathbf{0}_{p/2},\mathbf{-3}_{p/2})^\T,\Ib)$ with Gaussian noise and uniform outliers with entry-wise contamination.
The left panel shows the HA Rand index and the right panel collects the estimation error. 
In all panels, purple, red, blue, and black lines mark our proposed method with Gaussian kernel weights, least-squares convex clustering with Gaussian kernel weights, $k$-means, and hierarchical clustering respectively.  
}
\label{overlap-clusters-1}
\end{figure}

\begin{figure}[!t]
\centering
\subfigure[varying sample sizes, $p = 20$, row-wise contamination = $10\%$, $\tau=0.001, \phi=0.1$.]{
\includegraphics[height=0.18\textheight,width=.45\textwidth]{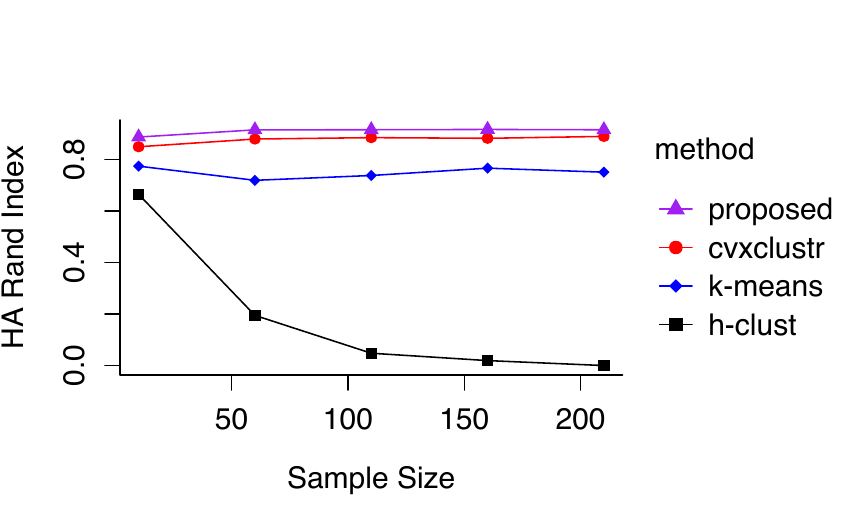}
\hspace{3mm}
\includegraphics[height=0.18\textheight,width=.45\textwidth]{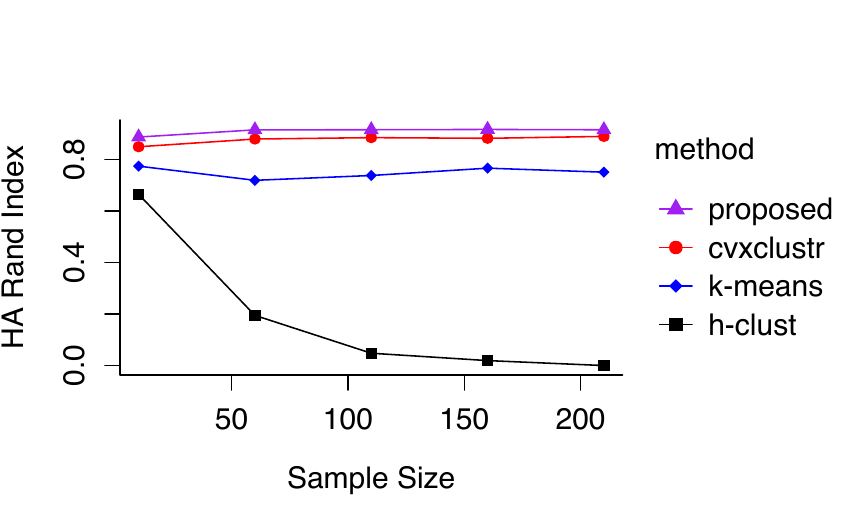}
}
\subfigure[varying feature dimensions, $n = 40$, row-wise contamination = $10\%$, $\tau=0.001, \phi=0.1$.]{
\includegraphics[height=0.18\textheight,width=.45\textwidth]{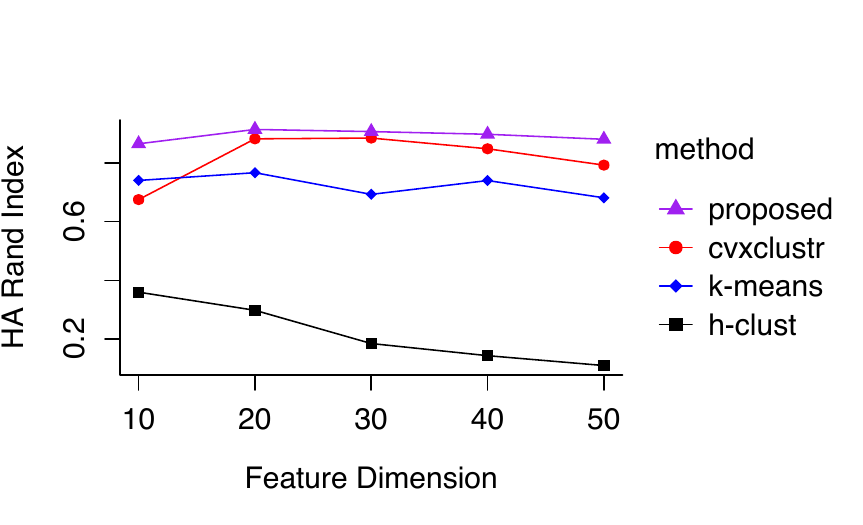}
\hspace{3mm}
\includegraphics[height=0.18\textheight,width=.45\textwidth]{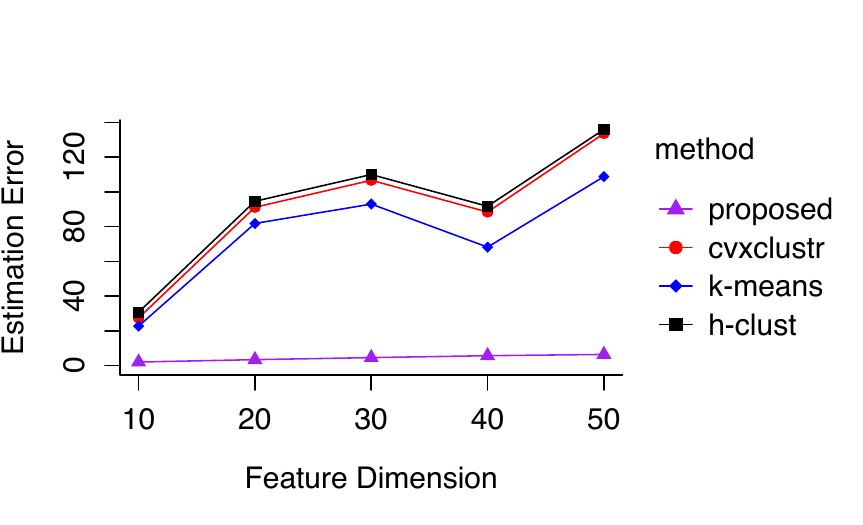}
}
\subfigure[varying row-wise outlier proportions, $n = 40, p = 20$, $\tau=0.001, \phi=0.1$.]{
\includegraphics[height=0.18\textheight,width=.45\textwidth]{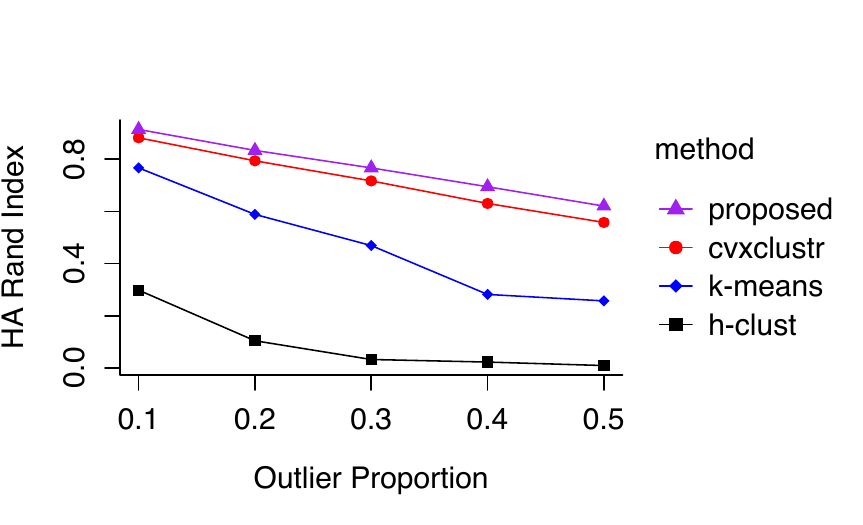}
\hspace{3mm}
\includegraphics[height=0.18\textheight,width=.45\textwidth]{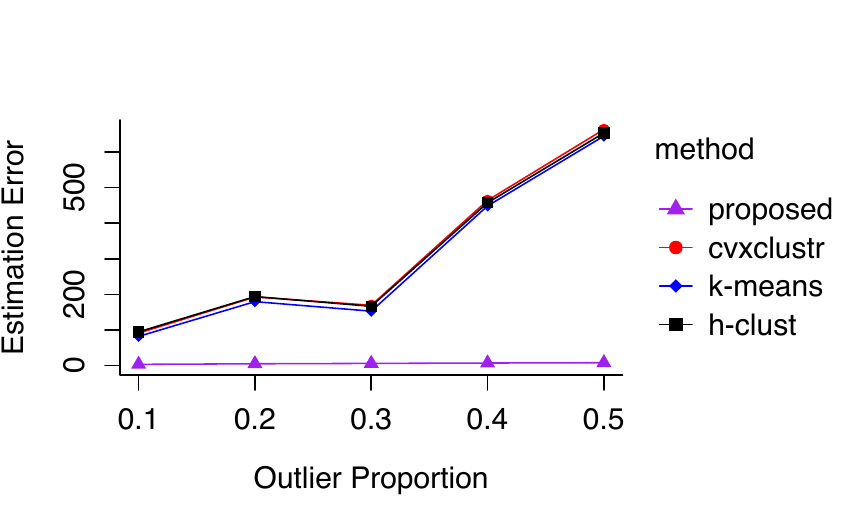}
}
\caption{Investigating the effect of overlapping clusters. 
Data generated from clusters $\Ub_1 \sim \mathcal{N}_p(\mathbf{0},\Ib)$ and $\Ub_2 \sim \mathcal{N}_p((\mathbf{0}_{p/2},\mathbf{-1}_{p/2})^\T,\Ib)$ with Gaussian noise and $t$-outliers with 1 degree of freedom and row-wise contamination.
The left panel shows the HA Rand index and the right panel collects the estimation error. 
In all panels, purple, red, blue, and black lines mark our proposed method with Gaussian kernel weights, least-squares convex clustering with Gaussian kernel weights, $k$-means, and hierarchical clustering respectively.  
}
\label{overlap-clusters-2}
\end{figure}

\begin{figure}[!t]
\centering
\includegraphics[width=.49\textwidth]{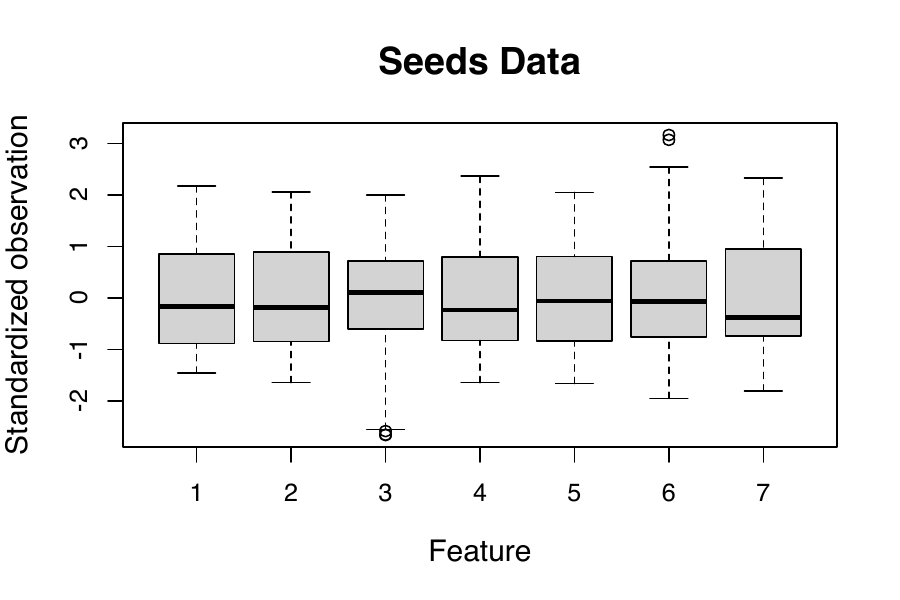}
\includegraphics[width=.49\textwidth]{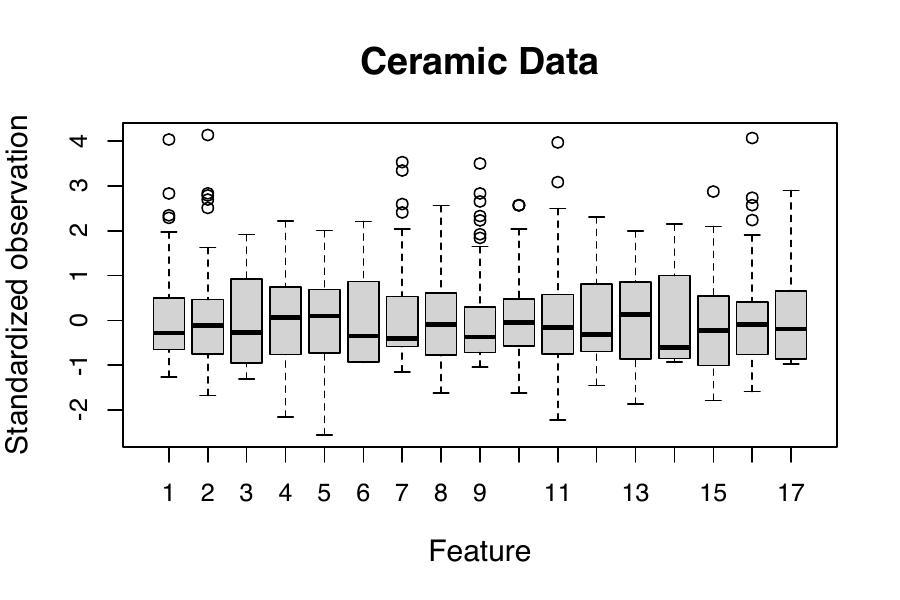} \\ 
\includegraphics[width=.98\textwidth]{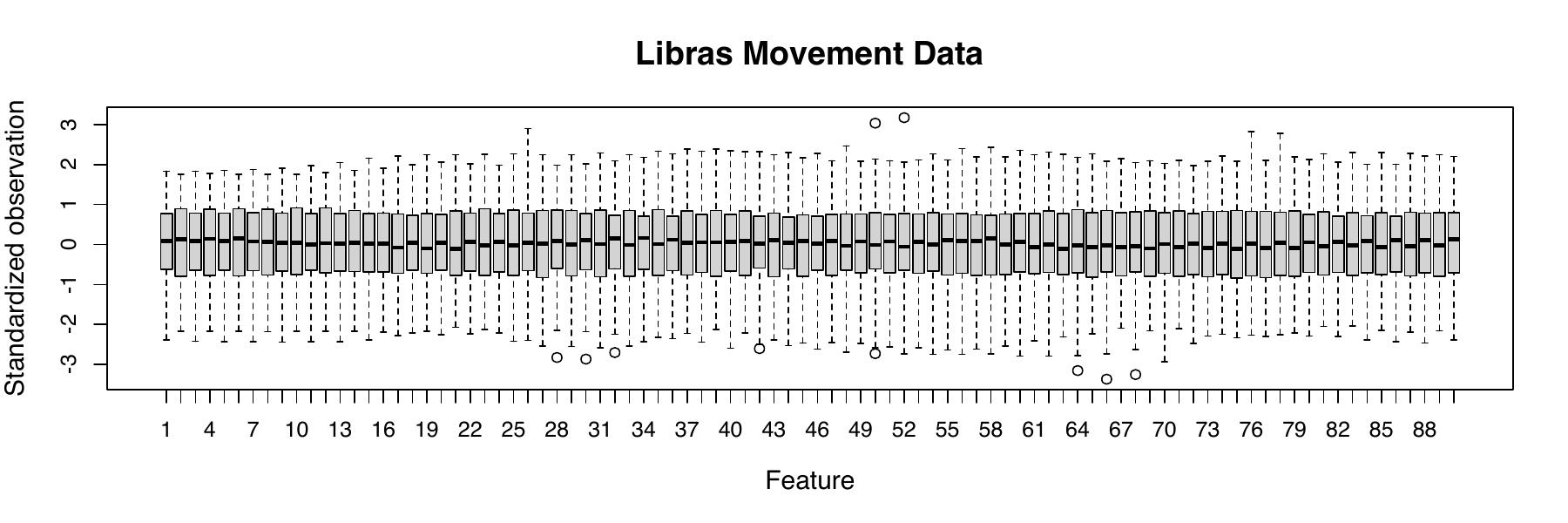}
\caption{Boxplots for three real-world data sets from UCI Machine Learning Repository.}
\label{real-boxplots}
\end{figure}





\newpage
\appendix 
\renewcommand{\theequation}{S.\arabic{equation}}
\renewcommand{\thetable}{S.\arabic{table}}
\renewcommand{\thefigure}{S.\arabic{figure}}
\renewcommand{\thesection}{S.\arabic{section}}
\renewcommand{\thelemma}{S.\arabic{lemma}}

\vspace{30pt}
\noindent{\bf \LARGE Appendix}

\paragraph{Norms.} We first provide explicit definitions for the norms we use in the paper. 
For a generic $p$-dimensional vector $\boldsymbol{v}=(v_1, \ldots, v_p)$, the $\ell_q$ norm for $q\geq1$ is defined as $\|\boldsymbol{v}\|_q=(\sum_{i=1}^p |v_i|^q)^{1/q}$.  
For a generic $n\times p$ matrix $\Ab$ with $a_{ij}$ being its $(i,j)$-th entry, the Frobenius norm is defined as $\|\Ab\|_{\rF}=\sqrt{\sum_{i=1}^n \sum_{j=1}^p a_{ij}^2}$.

\paragraph{HA Rand Index.} We also give a precise definition of the HA Rand index we use as a measure of clustering accuracy. 
Given a set of $n$ elements $O=\{o_1, \ldots, o_n\}$, consider two partitions of $O$: $X = \{X_1, X_2, \ldots, X_r\}$ that clusters $O$ into $r$ subsets and $Y = \{Y_1, Y_2, \ldots, Y_s\}$ that clusters $O$ into $s$ subsets, where each $X_i$ and $Y_j$ are the clusters in the two partitions.
The HA Rand Index is calculated using the following formula:
\[
\frac{\binom{n}{2} \sum_{i,j} \binom{n_{ij}}{2} - \sum_i \binom{n_{i\cdot}}{2} \sum_j \binom{n_{\cdot j}}{2}}{\frac{1}{2} \binom{n}{2} \left[ \sum_i \binom{n_{i\cdot}}{2} + \sum_j \binom{n_{\cdot j}}{2} \right] - \sum_i \binom{n_{i\cdot}}{2} \sum_j \binom{n_{\cdot j}}{2}},
\]
where $n_{ij}$ is the number of elements that are in both $X_i$ and $Y_j$; 
$n_{i\cdot}=\sum_j n_{ij}$ is the number of elements in $X_i$; $n_{\cdot j}=\sum_i n_{ij}$ is the number of elements in $Y_j$.
The value of the HA Rand index ranges from $-1$ to $1$, with a value close to $1$ indicating strong agreement between the true and estimated clusters.

\section{Derivation of Algorithm~\ref{Alg:huberadmm}}\label{app:1}
We give some details on Algorithm \ref{Alg:huberadmm}. Recall that the scaled augmented Lagrangian function for \eqref{eq:ccl4} takes the form
\begin{align*}
L_{\tau}(\Wb,\Vb,\Ub,\Yb,\Zb)&=
\sum_{i=1}^n \ell_\tau (\Xb_{i}- \Wb_{i})  + \lambda \sum_{i<i'}w_{ii'}\|\Vb_{ii'} \|_2\\
&+ \sum_{i<i' }\frac{\rho}{2} \|\Vb_{ii'}  - (\Ub_{i} - \Ub_{i'}) + \Yb_{ii'}    \|_2^2+ \frac{\rho}{2} \|\Wb  - \Ub + \Zb    \|_{\rF}^2
\end{align*}
The alternating direction method of multipliers algorithm requires the following updates:
\begin{align*}
\Wb^{(t+1)}&= \underset{\Wb}{\mathrm{argmin}}\;\mathit{L}_{\tau}(\Wb,\Vb^{(t)},\Ub^{(t)},\Yb^{(t)},\Zb^{(t)});\\
\Vb^{(t)}&= \underset{\Vb}{\mathrm{argmin}}\;\mathit{L}_{\tau}(\Wb^{(t+1)},\Vb,\Ub^{(t)},\Yb^{(t)},\Zb^{(t)});\\
\Ub^{(t+1)}&= \underset{\Ub}{\mathrm{argmin}}\;\mathit{L}_{\tau}(\Wb^{(t+1)},\Vb^{(t+1)},\Ub,\Yb^{(t)},\Zb^{(t)});\\
\Yb_{ii'}^{(t+1)}&= \Yb_{ii'}^{(t)}-\rho (\Ub^{(t+1)}_{i}-\Ub^{(t+1)}_{i'}-\Vb^{(t+1)}_{ii'});\\
\Zb^{(t+1)}&= \Zb^{(t)}-\rho(\Ub^{(t+1)}-\Wb^{(t+1)}).
\end{align*}
We now derive the updates for $\Wb$, $\Vb$, and $\Ub$.

\textbf{Update for $\Wb$:}
An update for $\Wb$ can be obtained by solving the following minimization problem:
\begin{equation*}
\underset{\Wb}{\mathrm{minimize}}~\sum_{i=1}^n \ell_\tau (\Xb_{i}- \Wb_{i}) +\frac{\rho}{2}\|\Wb-\Ub+\Zb\|_{\rF}^{2}.
\end{equation*}
The above problem can be solved element-wise:
\begin{equation}
\label{w1}
\underset{W_{ij}}{\mathrm{minimize}}\;l_{\tau}(X_{ij}-W_{ij})+\frac{\rho}{2}(W_{ij}-U_{ij}+Z_{ij})^2
\end{equation}
Due to the use of Huber loss, there are two different cases: (i) $|X_{ij}-W_{ij}| \leq \tau$; and (ii) $|X_{ij}-W_{ij}| > \tau$.  

For the case when $|X_{ij}-W_{ij}| \leq \tau$, \eqref{w1} reduces to
\begin{equation*}
\underset{W_{ij}}{\mathrm{minimize}}~ \frac{1}{2}(X_{ij}-W_{ij})^{2}+\frac{\rho}{2}(W_{ij}-U_{ij}+Z_{ij})^2.
\end{equation*}
Thus, we have $\hat{W}_{ij}=\{X_{ij}+\rho(U_{ij}-Z_{ij})\}/(1+\rho)$.
Substituting this into the constraint $|X_{ij}-W_{ij}| \leq \tau$, we obtain $|\rho\{X_{ij}-(U_{ij}-Z_{ij})\}/(1+\rho)|\leq \tau$. 
Thus,
\[
\hat{W}_{ij}=\{X_{ij}+\rho(U_{ij}-Z_{ij})\}/(1+\rho),\qquad  \mathrm{if}~ |\rho\{X_{ij}-(U_{ij}-Z_{ij})\}/(1+\rho)|\leq \tau.
\]

For the case $|X_{ij}-W_{ij}| > \tau$, we solve the problem
\begin{equation*}
\underset{W_{ij}}{\mathrm{minimize}}~\tau|X_{ij}-W_{ij}|+\frac{\rho}{2}(W_{ij}-U_{ij}+Z_{ij})^2.
\end{equation*}
To this end, let $H_{ij}=X_{ij}-W_{ij}$. By a change of variable, we have
\begin{equation*}
\underset{H_{ij}}{\mathrm{minimize}}~\frac{\tau}{\rho}|H_{ij}|+\frac{1}{2}(X_{ij}-H_{ij}-U_{ij}+Z_{ij})^2
\end{equation*}
It can be shown that $\hat{H}_{ij}=S\{X_{ij}-(U_{ij}-Z_{ij}),\tau/\rho\}$, where
$S(a,b)= \mathrm{sign}(a) \max (|a|-b,0)$ is the soft-thresholding operator. 
Thus we have 
\begin{equation*}
\hat{W}_{ij}=X_{ij}-S\{X_{ij}-(U_{ij}-Z_{ij}),\tau/\rho\},\qquad  \mathrm{if}~ |\rho\{X_{ij}-(U_{ij}-Z_{ij})\}/(1+\rho)|> \tau.
\end{equation*}

\textbf{Update for $\Vb$:}
For each pair of $i<i'$, we update $\Vb_{ii'}$ by solving the problem:
\begin{equation*}
\underset{\Vb_{ii'}}{\mathrm{minimize}}~ \frac{\lambda w_{ii'}}{\rho} \|\Vb_{ii'}\|_{2}+\frac{1}{2}\|\Vb_{ii'}-(\Ub_{i}-\Ub_{i'})+\Yb_{ii'}\|_{2}^{2}.
\end{equation*}
This is a standard group lasso problem with the following update:
\begin{equation*}
\hat{\Vb}_{ii'}= \left[1-\frac{\lambda w_{ii'}}{\rho \|\Ub_{i}-\Ub_{i'}-\Yb_{ii'} \|_2}\right]_+ (\Ub_{i}-\Ub_{i'}-\Yb_{ii'}),
\end{equation*}
where $[ a]_+ = \max(0,a)$.

\textbf{Update for $\Ub$:}
To update $\Ub$, we solve 
\begin{equation}
\label{u1}
\underset{\Ub \in \RR^{n\times p}}{\mathrm{minimize}}~ \sum_{i<i'}\frac{\rho}{2}\|\Vb_{ii'}-(\Ub_{i}-\Ub_{i'})+\Yb_{ii'}\|_{2}^{2}+\frac{\rho}{2}\|\Wb-\Ub+\Zb\|_{\rF}^{2}.
\end{equation}
To simplify the expression above, we construct an ${n\choose 2}\times n$ matrix $\Eb$ such that $(\Eb \Ub)_{ii'} = \Ub_i-\Ub_{i'}$. 
Then, \eqref{u1} is equivalent to
\begin{equation*}
\underset{\Ub\in \RR^{n\times p}}{\mathrm{minimize}}~\|\Vb+\Yb-\Eb\Ub\|_{\rF}^2 + \|\Wb-\Ub+\Zb\|_{\rF}^2.
\end{equation*}
Solving the above yields
\begin{equation*}
\hat{\Ub}=(\Eb^{\T}\Eb+\Ib)^{-1}\{\Eb^{\T}(\Vb+\Yb)+(\Wb+\Zb)\}.
\end{equation*}

\section{Proof of Theorem~\ref{thm:ls} }\label{appendix:b0}

We construct a contaminated data $\widetilde \Xb$ with only the $(1,1)$-th entry contaminated.  Specifically, for $\forall M>1$, let $\widetilde X_{11}=X_{11}+M$  with all other entries in $\tilde{\mathbf X}$ being the same with that of $\mathbf X$, and we will have $\widehat U^{\rm ls}_{11}(\widetilde \Xb)=\widetilde X_{11}=X_{11}+M$, resulting in $\|\hat {\mathbf{U}}^{\rm ls}(\tilde{\mathbf X}) -\hat {\mathbf{U}}^{\rm ls}({\mathbf X}) \|_\rF\geq M$.

By the optimality conditions, we have 
\$
(\widehat\Ub_1^{\rm ls} - \tilde \Xb_1) - \lambda\sum_{j=2}^n \zb^{(1,j)} = 0, 
\$
where 
$
\zb^{(1,j)}\in \partial \|\widehat\Ub_1^{\rm ls}-\widehat\Ub_j^{\rm ls}\|_2
$
and 
\$
\partial \|\widehat\Ub_1^{\rm ls}-\widehat\Ub_j^{\rm ls}\|_2 
=\begin{cases}
\frac{\widehat\Ub_1^{\rm ls}-\widehat\Ub_j^{\rm ls}}{\|\widehat\Ub_1^{\rm ls}-\widehat\Ub_j^{\rm ls}\|_2} & {\rm if} ~\widehat\Ub_1^{\rm ls}-\widehat\Ub_j^{\rm ls}\ne 0, \\
\{\Vb\in \RR^p : \|\Vb\|_2\leq 1\} &  {\rm if} ~\widehat\Ub_1^{\rm ls}-\widehat\Ub_j^{\rm ls}= 0. 
 \end{cases}
\$
Observe that, for any $\Vb\in \partial \|\widehat\Ub_1^{\rm ls}-\widehat\Ub_j^{\rm ls}\|_2 $, we have $\|\Vb\|_2\leq 1$.  Thus we have 
\$
\|\widehat\Ub_1^{\rm ls} - \tilde \Xb_1\|_2 =\| \lambda\sum_{j=2}^n \zb^{(1,j)}\|_2\leq \lambda (n-1)<\infty.
\$
Now if $\widehat\Ub^{\rm ls}$ is bounded, and thus  is $\widehat\Ub_1^{\rm ls}$. We have 
\$
 M -\|\Xb_1\|_2 - \|\widehat\Ub_1^{\rm ls}\|_2 \leq \|\tilde \Xb_1\|_2 - \|\widehat\Ub_1^{\rm ls}\|_2 \leq \lambda (n-1)<\infty.
\$
Taking $M\rightarrow \infty$, we arrive at a contradiction. Thus $\widehat\Ub^{\rm ls}$ must be unbounded. Thus the breakdown point is $1/n$. 
\section{Proof of Theorem~\ref{thm:bp} (upper bound part)}
\begin{lemma}
\label{lemma:bound}
For any two scalars $s$ and $t$, we have
$
\ell_\tau(s+t)\leq \ell_{\tau}(s)+\ell_\tau(t)+\tau^2
$.
\end{lemma}
\begin{proof}[Proof of Lemma \ref{lemma:bound}]
\label{appendix:d}
We break the proof into $6$ cases. 

Case 1. Suppose that $|s+t|\leq \tau$, $|s|\leq \tau$ and $|t|\leq \tau$. In this case, we have
\$
\ell_\tau(s+t)&=\frac{1}{2}(s+t)^2= \frac{1}{2}s^2+\frac{1}{2}t^2+st\\
&\leq \frac{1}{2}s_2^2+\frac{1}{2}t^2+\tau^2=\ell_{\tau}(s)+\ell_\tau(t)+\tau^2.
\$

Case 2. Suppose that $|s+t|\leq \tau$, $|s|\geq \tau$ and $|t|\leq \tau$. In this case, we must have $st\leq 0$ and thus  
\$
\ell_\tau(s+t)&=\frac{1}{2}(s+t)^2=\frac{1}{2}s^2+\frac{1}{2}t^2+st\\
&\leq \frac{1}{2}s^2+\frac{1}{2}t^2=\ell_{\tau}(s)+\ell_\tau(t).
\$

Case 3. Suppose that $|s+t|\leq \tau$, $|s|\geq \tau$ and $|t|\geq \tau$. Similar to Case 2, we must have $st\leq 0$ and thus 
\$
\ell_\tau(s+t)\leq\ell_{\tau}(s)+\ell_\tau(t).
\$

Case 4. Suppose that $|s+t|> \tau$, $|s|\leq \tau$ and $|t|\leq \tau$. In this case, we must have $st\geq 0$. Without loss of generality, we assume $0<s,\ t<\tau$. Therefore, we have
\$
\ell_\tau(s+t)&=\tau|s+t|-\frac{1}{2}\tau^2= \tau(s+t)-\frac{1}{2}\tau^2 \\
&\leq \frac{1}{2}s^2+\frac{1}{2}t^2+\frac{1}{2}\tau^2=\ell_{\tau}(s)+\ell_\tau(t)+\frac{1}{2}\tau^2,
\$
where the last inequality is due to the fact that $s^2+t^2-2\tau(s+t)+2\tau^2\geq 0.$

Case 5. Suppose that $|s+t|> \tau$, $|s|\leq \tau$ and $|t|> \tau$. In this case, we have
\$
\ell_\tau(s+t)&=\tau|s+t|-\frac{1}{2}\tau^2\leq \frac{1}{2}s^2+\tau|t|-\frac{1}{2}\tau^2+\tau^2 \\
&=\ell_{\tau}(s)+\ell_\tau(t)+\tau^2.
\$

Case 6. Suppose that $|s+t|> \tau$, $|s|> \tau$ and $|t|> \tau$. In this case, we have
\$
\ell_\tau(s+t)&=\tau|s+t|-\frac{1}{2}\tau^2\leq \tau|2|-\frac{1}{2}\tau^2+\tau|t|-\frac{1}{2}\tau^2+\frac{1}{2}\tau^2 \\
&=\ell_{\tau}(s)+\ell_\tau(t)+\frac{1}{2}\tau^2.
\$

Combining all the results above in different cases completes the proof. 
\end{proof}
\begin{lemma}\label{lemma:bound2}
For any $u_k\rightarrow u>0$ and $t_k\rightarrow \infty$, we have 
$
\lim_{k\rightarrow \infty} {\ell_\tau(t_ku_k)}/{\ell_\tau(t_k)}=u.
$
\end{lemma}
\begin{proof}[Proof of Lemma \ref{lemma:bound2}]
The proof of this lemma is a direct application of L'Hopital's Rule and thus is omitted. 
\end{proof}
\begin{lemma}\label{lemma:bound3}
For any $u_k\rightarrow u>0$ and $t_k\rightarrow \infty$, we have 
$
\lim_{k\rightarrow \infty} {t_ku_k}/{\ell_\tau(t_k)}=u/\tau.
$
\end{lemma}
\begin{proof}[Proof of Lemma \ref{lemma:bound3}]
The proof of this lemma is a direct application of L'Hopital's Rule and thus is omitted. 
\end{proof}
\label{appendix:b}
Here we give the proof of the upper bound in Theorem~\ref{thm:bp}.

\begin{proof}
An upper bound for the breakdown point. We now sharpen the upper bound in \eqref{bdp:upper}. First, the cost function in \eqref{eq:ccl4} is translation invariant with respect to translation of $\mathbf{1}\ab^\T\in \RR^{n\times p}$ which indicates the obtained estimator is also translation equivariant, i.e.,
\$
\hat\Ub (\Zb+\mathbf{1}\ab^\T)=\hat\Ub(\Zb)+\mathbf{1}\ab^\T
\$
for any data matrix $\Zb$ and any vector $\ab\in \RR^p$. Take $\ab =\eta {\bf 1}$ and let $\Xb^{1,\eta}$ be a data matrix such that the rows are
\$
\left\{\Xb_1,\ldots, \Xb_{n-\lfloor (n+1)/2\rfloor}, \Xb_{n-\lfloor (n+1)/2\rfloor+1}+\eta {\bf 1}, \ldots \Xb_{n}+\eta{\bf 1}\right\},
\$
where ${\bf 1}$ is a vector of all $1$'s. 
Because there are $\lfloor (n+1)/2\rfloor$ 	contaminated rows,  $\Xb^{1,\eta}\in \cP_{\lfloor (n+1)/2\rfloor}(\Xb)$ for any $\eta$. Similarly, let $\Xb^{2,\eta}$ be a data matrix such that the rows are
\$
\left\{\Xb_1-\eta{\bf 1},,\ldots, \Xb_{n-\lfloor (n+1)/2\rfloor}-\eta{\bf 1},, \Xb_{n-\lfloor (n+1)/2\rfloor+1}, \ldots \Xb_{n}\right\}.
\$
Because $\Xb^{2,\eta}$ has $n-\lfloor(n+1)/2\rfloor$ contaminated rows where $n-\lfloor(n+1)/2\rfloor\leq \lfloor(n+1)/2\rfloor$, $\Xb^{2,\eta}\in \cP_{\lfloor(n+1)/2\rfloor}(\Xb)$ for any $\eta$. Moreover, we have $\Xb^{1,\eta}=\Xb^{2,\eta}+\eta{\bf 1}$, where ${\bf 1}$ is a matrix of all $1$'s, with some abuse of notation. 

By triangle inequality, we have 
\$
\left\|\hat\Ub(\Xb^{1,\eta})-\hat\Ub(\Xb^{2,\eta})\right\|_\rF\leq \left\|\hat\Ub(\Xb^{1,\eta})-\hat\Ub(\Xb)\right\|_\rF+\left\|\hat\Ub(\Xb)-\hat\Ub(\Xb^{2,\eta})\right\|_\rF, 
\$
which by translation inequality reduces to
\$
\eta \sqrt{n p} 
&\leq \left\|\hat\Ub(\Xb^{1,\eta})-\hat\Ub(\Xb)\right\|_\rF+\left\|\hat\Ub(\Xb^{2,\eta})-\hat\Ub(\Xb)\right\|_\rF\\
&\leq 2\sup_{\tilde\Xb\in \cP_{\lfloor (n+1)/2\rfloor}(\Xb)}\big\|\widehat\Ub (\tilde\Xb)-\hat{\Ub} (\Xb)\big\|_\rF.
\$ 
Taking $\eta\rightarrow \infty$ acquires 
\$
\sup_{\tilde\Xb\in \cP_{\lfloor (n+1)/2\rfloor}(\Xb)}\big\|\widehat\Ub (\tilde\Xb)-\hat{\Ub} (\Xb)\big\|_\rF=\infty.
\$
This implies 
\$
\frac{m}{n}\leq \frac{\lfloor (n+1)/2\rfloor}{n}.
\$

Therefore, combining the results in both parts gives 
\$
\frac{1}{2}\leq\frac{m}{n}\leq \frac{\lfloor (n+1)/2\rfloor}{n},
\$
which holds under the assumption that
\$
\tau/\lambda\leq  \frac{n-\lfloor(n+1)/2\rfloor}{\sqrt{p}}.
\$

\end{proof}

\section{Proof of Proposition~\ref{prop:wo_penalty}}
\label{appendix:c}
We first take the derivative of the Huber loss function. We know that
\begin{equation}\nonumber
\nabla \ell_\tau(a)=\begin{cases}
a,&|a|\leq \tau\\
\sign(a)\tau,&|a|>\tau
\end{cases}
=\sign(a)\min(\tau,|a|),
\end{equation}
where $\sign(x)= -1$ if $x<0$, $\sign(x)= 1$ if $x>0$, $\sign(x)= 0$ if $x=0$.  
Therefore, the score function of the objective
$\sum_{i=1}^n\ell_\tau(\mathbf X_i-\mathbf U_i)$ is:
\begin{equation}\nonumber
\nabla\ell_\tau(X_{ij}-U_{ij})=-\text{sign}(X_{ij}-U_{ij})\min(\tau,|X_{ij}-U_{ij}|)=0,\forall i=1,...,n,\forall j=1,...,p
\end{equation}
which can be solved only by $\widehat U^{\rm wo}_{ij}(\Xb)=X_{ij}$,  the $(i,j)$-th entry of  $\widehat U^{\rm wo}_{ij}(\Xb)$, $\forall i,j$. 

For a fixed pair $(i,j)$, we construct a contaminated data $\widetilde \Xb$ with only the $(i,j)$-th entry contaminated.  Specifically, for $\forall M>1$, let $\widetilde X_{ij}=X_{ij}+M$  with all other entries in $\tilde{\mathbf X}$ being  exactly the same with that of $\mathbf X$, and we will have $\widehat U^{\rm wo}_{ij}(\widetilde \Xb)=\widetilde X_{ij}=X_{ij}+M$, resulting in $\|\hat {\mathbf{U}}^{\rm wo}(\tilde{\mathbf X}) -\hat {\mathbf{U}}^{\rm wo}({\mathbf X}) \|_\rF\geq M$. Taking $M\rightarrow \infty$, we obtain that $\| \hat{\mathbf U}^{\rm wo}(\tilde{\mathbf X})-\hat{\mathbf U}^{\rm wo}(\mathbf X)     \|_{\rF}\rightarrow \infty$, and hence the breakdown point of $\hat{\mathbf U}^{\rm wo}$ is $1/n$ if considering sample or row-wise, with respect to $\Xb$,  contamination; entry wise contamination. Additionally the breakdown point of $\hat{\mathbf U}^{\rm wo}$ is $1/(np)$ if considering entry-wise contamination.

\section{Hyper-parameters and additional numerical experiments}
We provide the hyper-parameters of the numeric experiments in Table~\ref{hyper}. 
\scolor{
The tolerance parameter $\epsilon$ is usually set to be small. The default choice of our algorithm is 1e-05, presented as tol\_abs in Table~\ref{hyper}. 
$\rho>0$ is the penalty parameter for the ADMM algorithm. 
Small $\rho$ would put less emphasis on maintaining the feasibility of the constraints, while large $\rho$ may put insufficient emphasis on minimizing the objective function. 
As noted in \citet{chi2013splitting}, their ADMM algorithm for convex clustering is guaranteed to converge for any $\rho>0$, which suggests that the algorithm is robust to the choice of $\rho$ to some extent. 
Therefore, to avoid choosing a $\rho$ that is too small or too large, we choose $\rho=1$ in our simulations and real-data analyses. 
}

\begin{table}[!t]
\centering
\caption{Hyper-parameters}
\label{hyper}
\resizebox{\textwidth}{!}{%
\begin{tabular}{lll}
\Xhline{1.2pt}
\textbf{Hyper-parameter} & \textbf{Interpretation}                                 & \textbf{Value} \\ \hline
lam.begin & the initial value of fusion penalty coefficient $\lambda$ & 0.01 \\
lam.step  & the increasing step length of   $\lambda$                 & 1.05 \\
max.log                  & the maximum number of iterations for $\lambda$                                  & 200    \\
rho                      & the nonnegative tuning parameter $\rho$ for the ADMM algorithm & 1              \\
tol\_abs                & the convergence tolerance level $\epsilon$                                      & 1e-05    \\
\Xhline{1.2pt}
\end{tabular}%
}
\end{table}

\scolor{ Figures~\ref{row-wise-1}--\ref{row-wise-3} provide the experiment results for row-wise contamination.}

\begin{figure}[!t]
\centering
\subfigure[varying sample sizes, $p=20$, row-wise contamination = $10\%$, $\tau=1$.]{
\includegraphics[height=0.18\textheight,width=.45\textwidth]{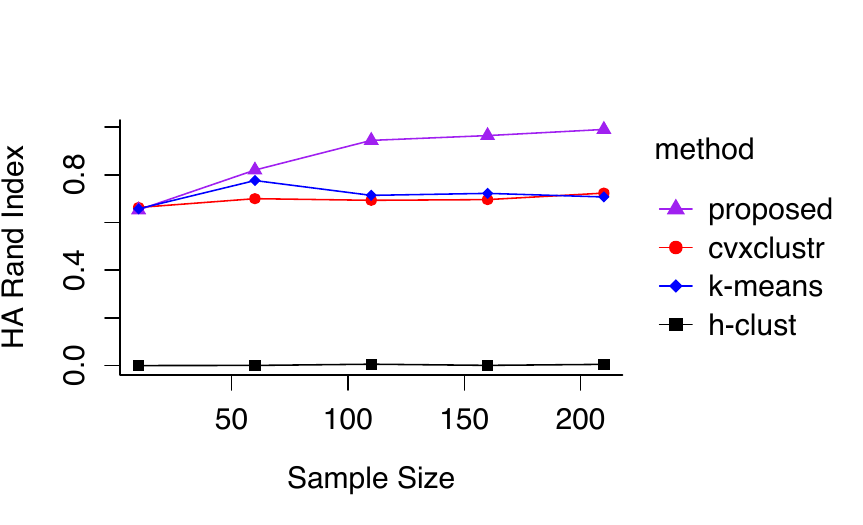}
\hspace{3mm}
\includegraphics[height=0.18\textheight,width=.45\textwidth]{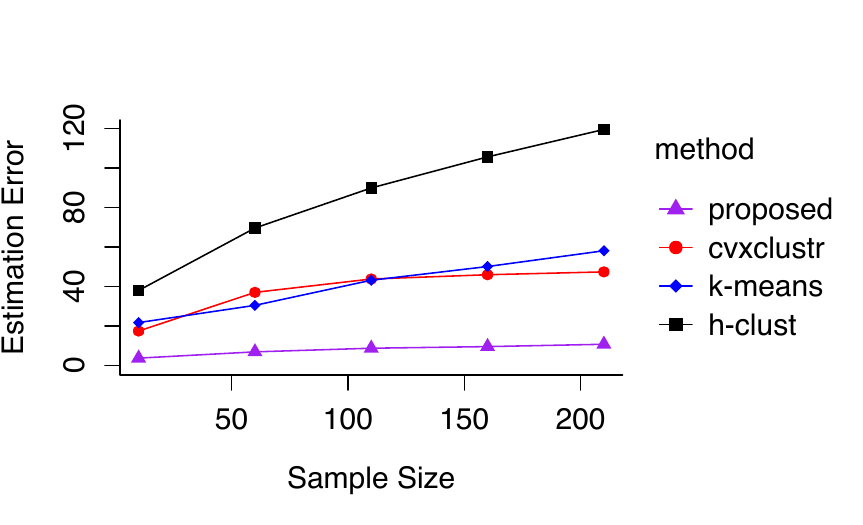}
}
\subfigure[varying feature dimensions, $n = 40$, row-wise contamination = $10\%$, $\tau=0.1$.]{
\includegraphics[height=0.18\textheight,width=.45\textwidth]{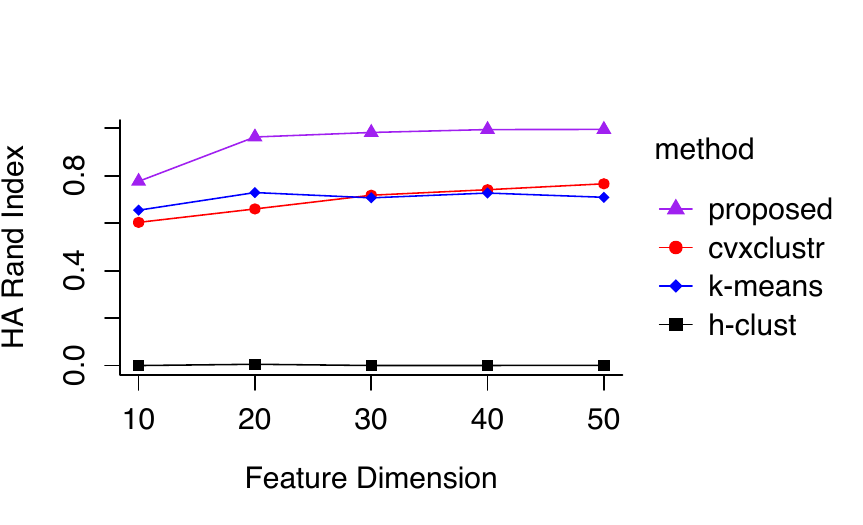}
\hspace{3mm}
\includegraphics[height=0.18\textheight,width=.45\textwidth]{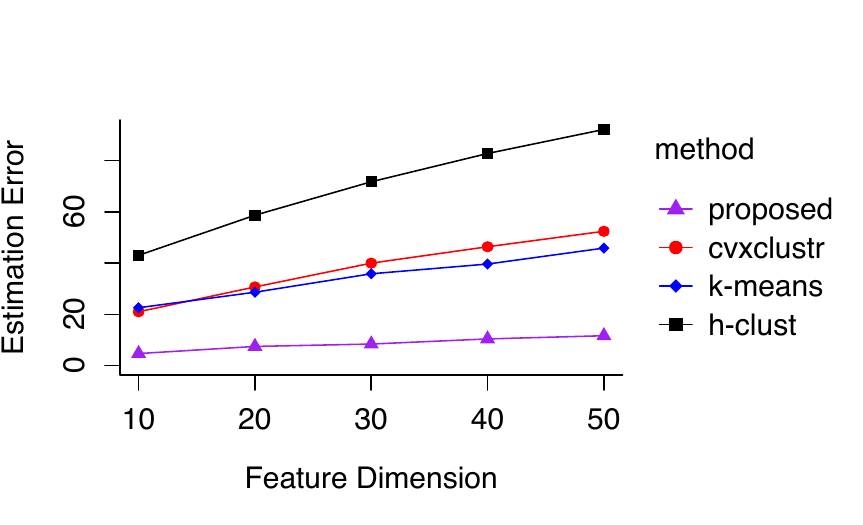}
}
\subfigure[varying sample sizes, $p=20$, row-wise contamination = $50\%$, $\tau=0.6$.]{
\includegraphics[height=0.18\textheight,width=.45\textwidth]{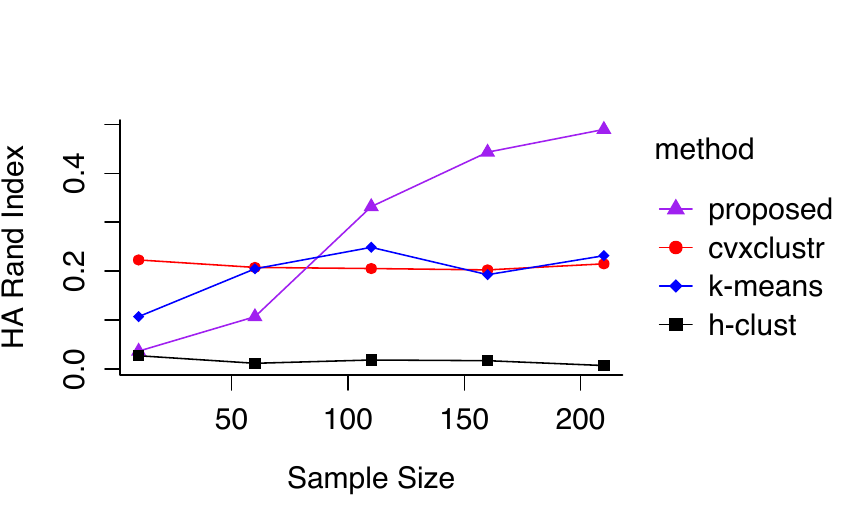}
\hspace{3mm}
\includegraphics[height=0.18\textheight,width=.45\textwidth]{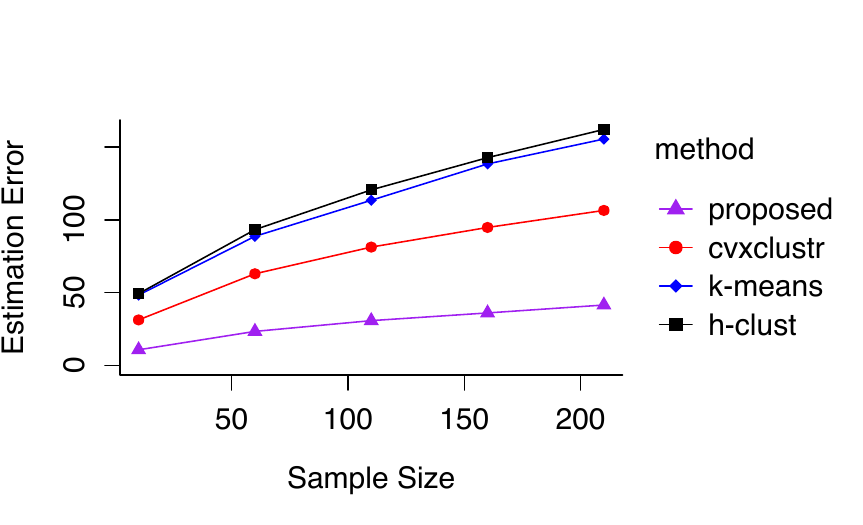}
}
\subfigure[varying feature dimensions, $n = 40$, row-wise contamination = $50\%$, $\tau=0.1$.]{
\includegraphics[height=0.18\textheight,width=.45\textwidth]{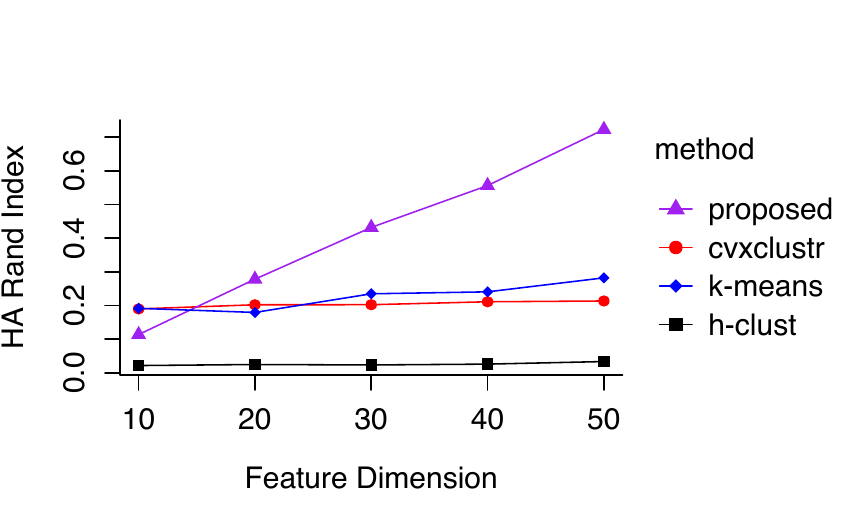}
\hspace{3mm}
\includegraphics[height=0.18\textheight,width=.45\textwidth]{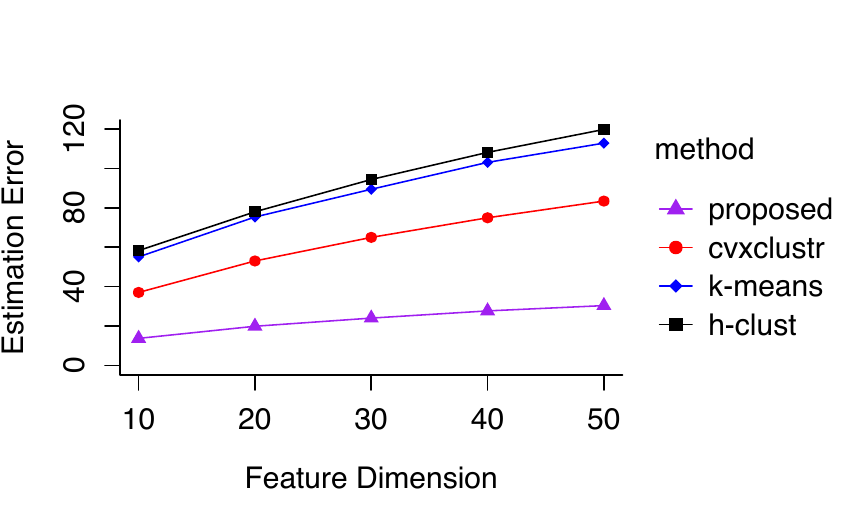}
}
\caption{Comparing our proposed method with others for data with Gaussian noise and uniform outliers with row-wise contamination. The left panel shows the HA Rand index and the right panel collects the estimation error. 
In all panels, purple, red, blue, and black lines mark our proposed method, least-squares convex clustering, $k$-means, and hierarchical clustering respectively. 
}
\label{row-wise-1}
\end{figure}

\begin{figure}[!t]
\centering
\subfigure[varying row-wise outlier proportions, $n = 20, p = 10$, $t$-noise with 5 degrees of freedom, $\tau=0.1$.]{
\includegraphics[height=0.18\textheight,width=.45\textwidth]{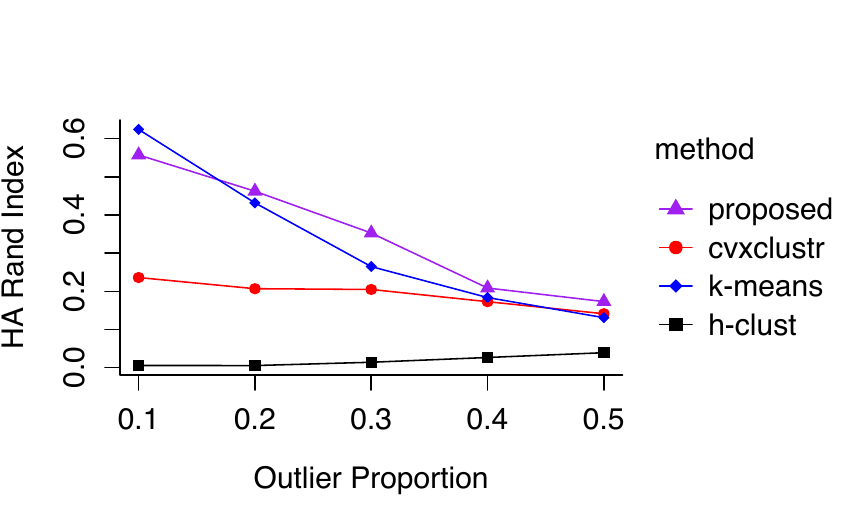}
\hspace{3mm}
\includegraphics[height=0.18\textheight,width=.45\textwidth]{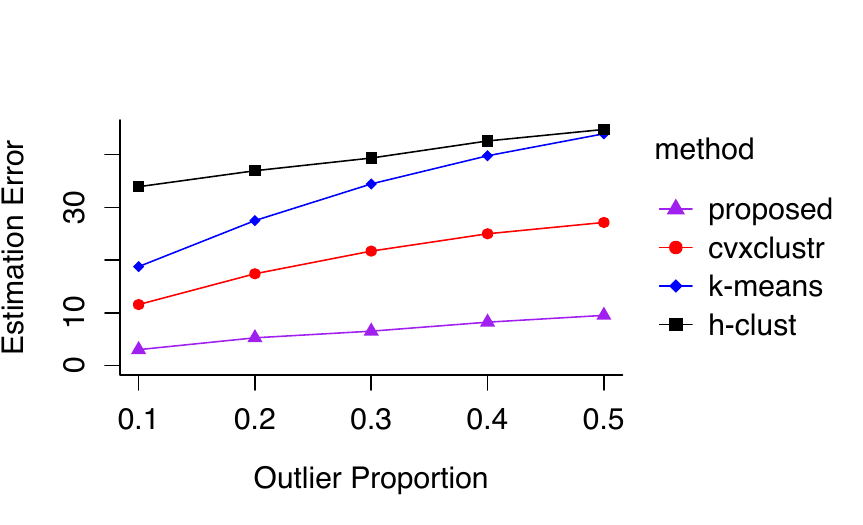}
}
\subfigure[varying row-wise outlier proportions, $n = 40, p = 20$, $t$-noise with 5 degrees of freedom, $\tau=0.1$.]{
\includegraphics[height=0.18\textheight,width=.45\textwidth]{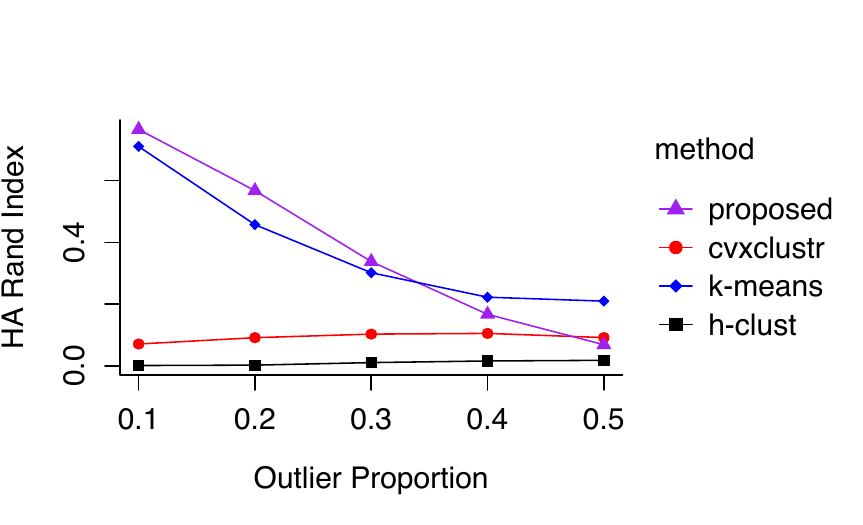}
\hspace{3mm}
\includegraphics[height=0.18\textheight,width=.45\textwidth]{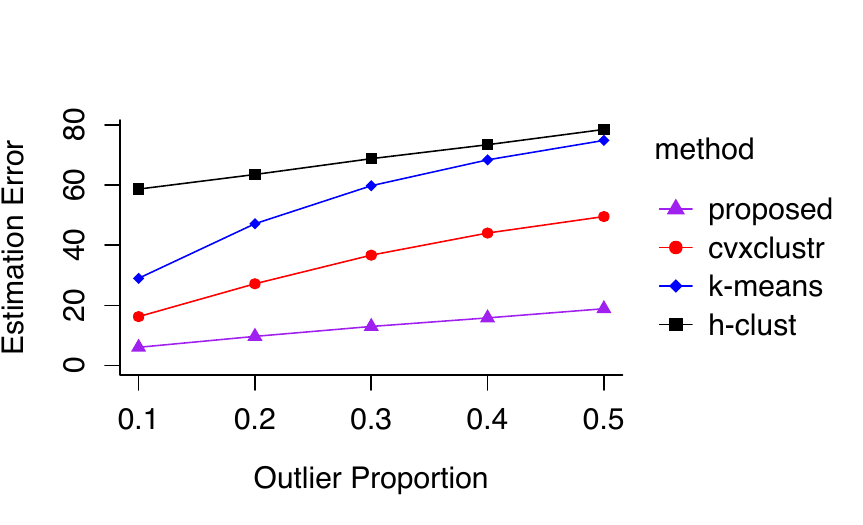}
}
\subfigure[varying degrees of freedom for $t$-noises, $n = 20, p=10$, row-wise contamination = $10\%$, $\tau=0.1$.]{
\includegraphics[height=0.18\textheight,width=.45\textwidth]{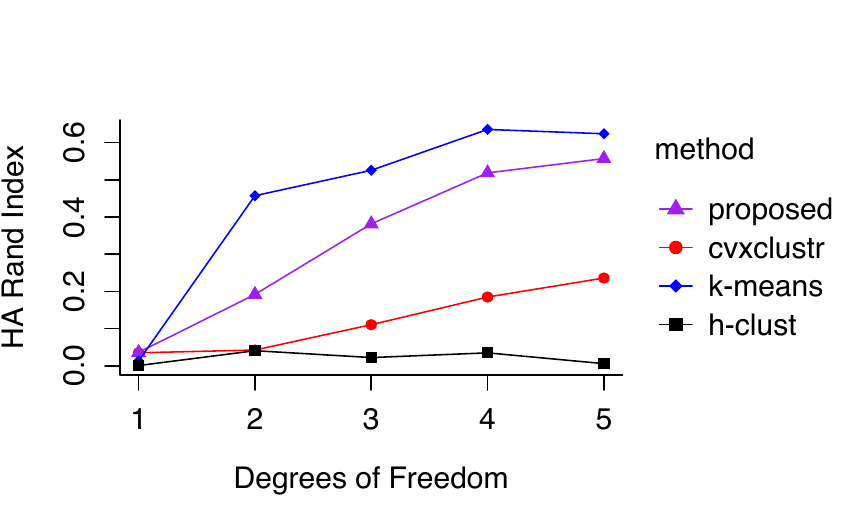}
\hspace{3mm}
\includegraphics[height=0.18\textheight,width=.45\textwidth]{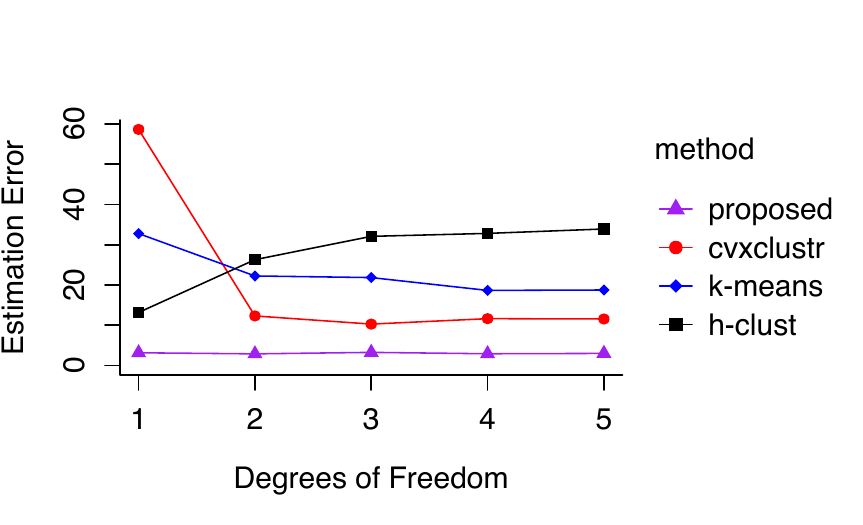}
}
\subfigure[varying degrees of freedom for $t$-noises, $n = 40, p=20$, row-wise contamination = $10\%$, $\tau=0.1$.]{
\includegraphics[height=0.18\textheight,width=.45\textwidth]{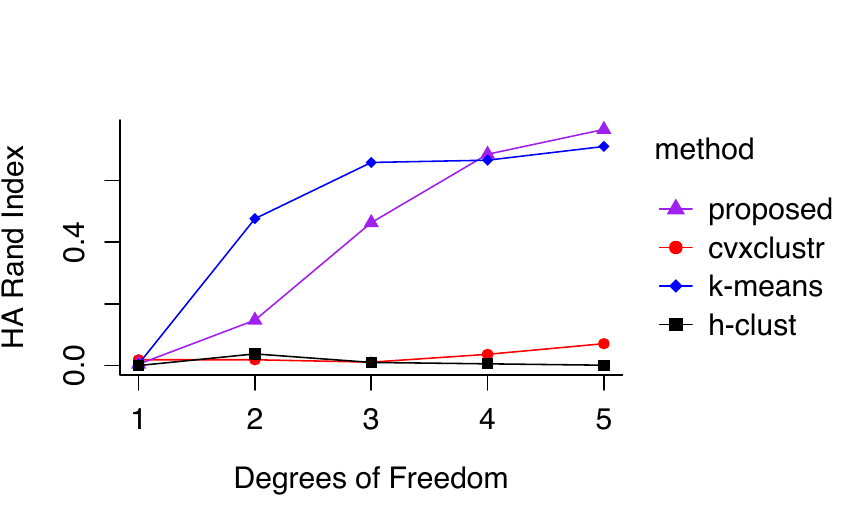}
\hspace{3mm}
\includegraphics[height=0.18\textheight,width=.45\textwidth]{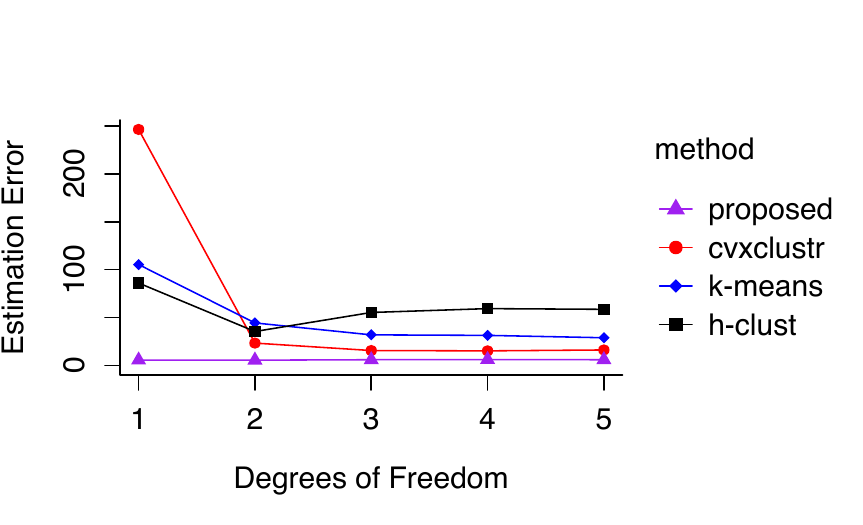}
}
\caption{Comparing our proposed method with others for data with $t$-noise and uniform outliers with row-wise contamination. 
The left panel shows the HA Rand index and the right panel collects the estimation error. 
In all panels, purple, red, blue, and black lines mark our proposed method, least-squares convex clustering, $k$-means, and hierarchical clustering respectively. 
}
\label{row-wise-2}
\end{figure}

\begin{figure}[!t]
\centering
\subfigure[varying sample sizes, $p = 20$, row-wise contamination = $10\%$, $\tau=1$.]{
\includegraphics[height=0.18\textheight,width=.45\textwidth]{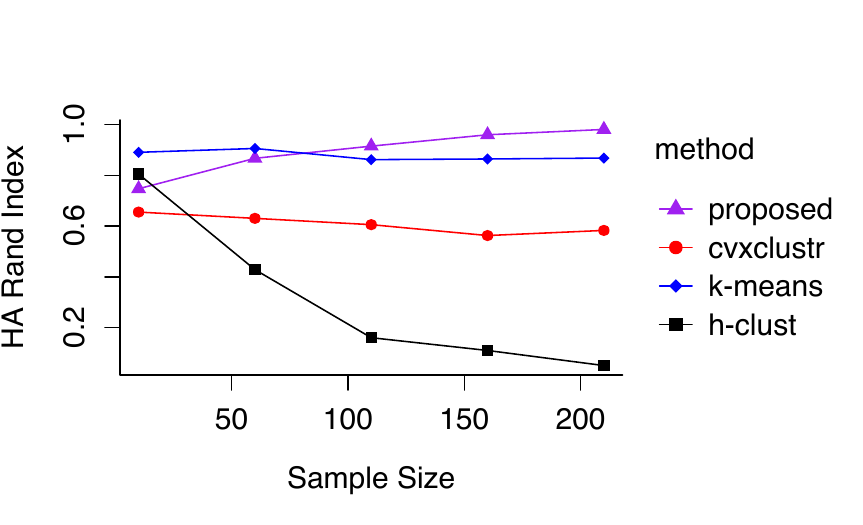}
\hspace{3mm}
\includegraphics[height=0.18\textheight,width=.45\textwidth]{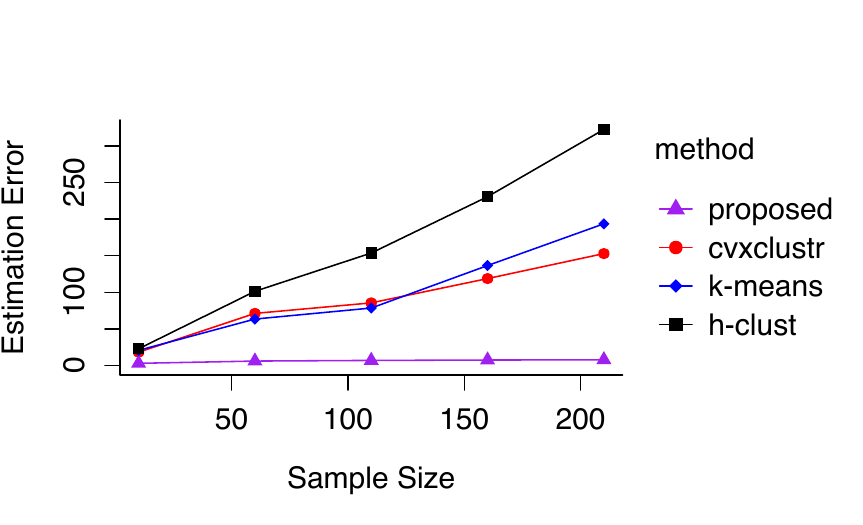}
}
\subfigure[varying feature dimensions, $n = 40$, row-wise contamination = $10\%$, $\tau=0.1$.]{
\includegraphics[height=0.18\textheight,width=.45\textwidth]{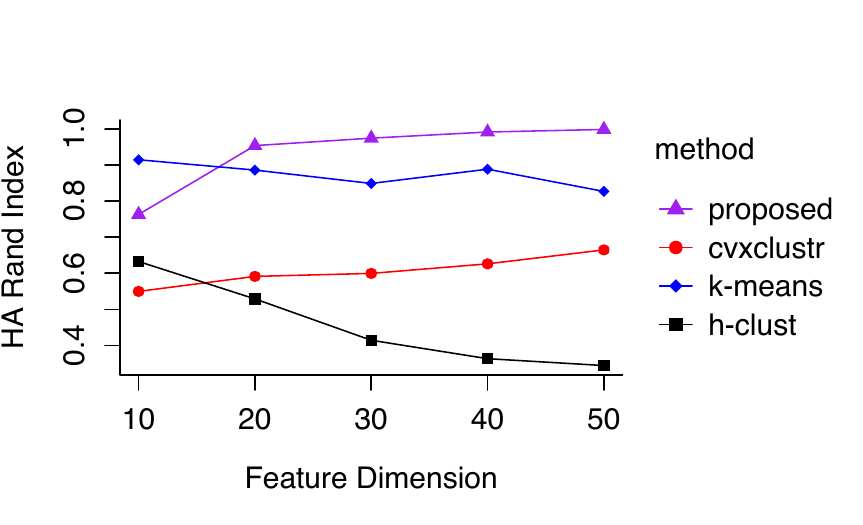}
\hspace{3mm}
\includegraphics[height=0.18\textheight,width=.45\textwidth]{pic/new_plots2/row-wise/normal_t_n40_varyp_r0.1_df1_tau0.1_error.pdf}
}
\subfigure[varying row-wise outlier proportions, $n = 40, p = 20$, $\tau=0.1$.]{
\includegraphics[height=0.18\textheight,width=.45\textwidth]{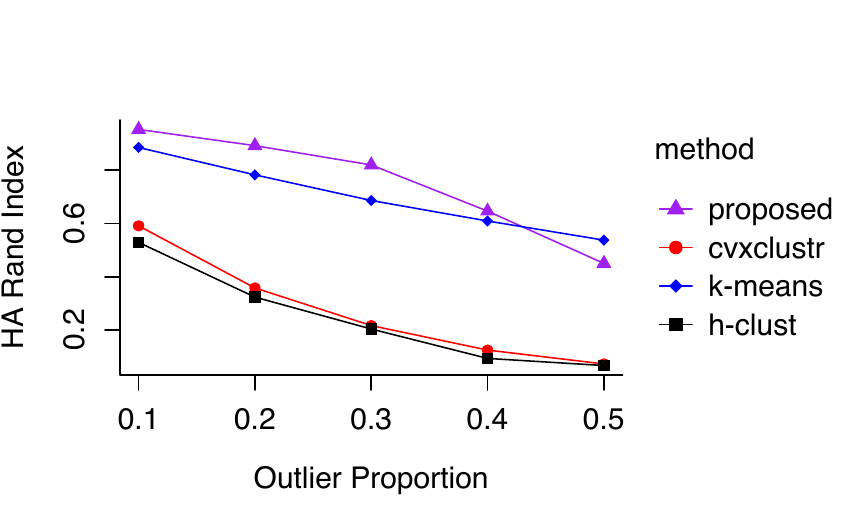}
\hspace{3mm}
\includegraphics[height=0.18\textheight,width=.45\textwidth]{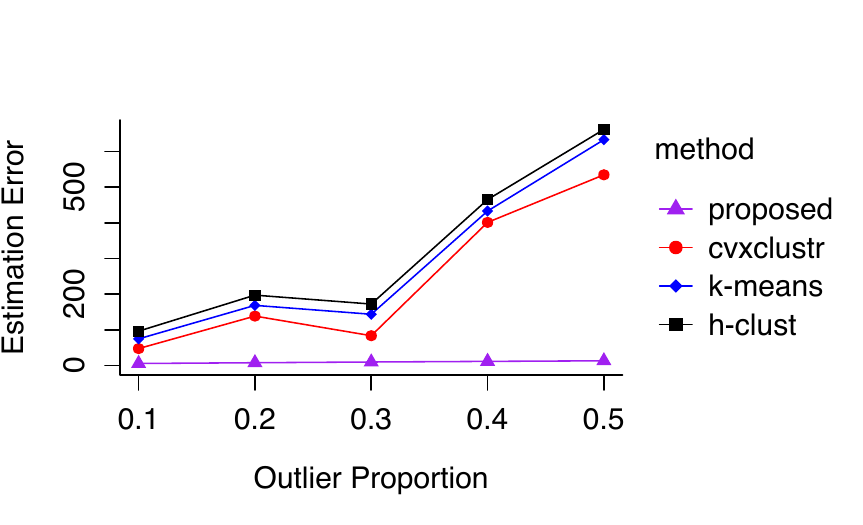}
}
\caption{Comparing our proposed method with others for data with Gaussian noise and $t$-outliers with 1 degree of freedom and row-wise contamination. The left panel shows the HA Rand index and the right panel collects the estimation error. 
In all panels, purple, red, blue, and black lines mark our proposed method, least-squares convex clustering, $k$-means, and hierarchical clustering respectively. 
}
\label{row-wise-3}
\end{figure}

\end{document}